\documentclass[12pt,a4paper]{article}

\usepackage[english]{babel}

\usepackage[a4paper,top=2cm,bottom=2cm,left=2.5cm,right=2.5cm,marginparwidth=1.75cm]{geometry}

\usepackage[style=authoryear-comp, backend=biber, sortcites=false]{biblatex}
\usepackage{amsmath}
\allowdisplaybreaks
\usepackage{amssymb}
\usepackage{amsthm}
\usepackage{multirow}
\usepackage{graphicx}
\usepackage[colorlinks=true, allcolors=blue]{hyperref}
\usepackage[title]{appendix}
\usepackage{mathrsfs}
\usepackage{mathtools}
\usepackage{amsfonts}
\usepackage{caption}
\usepackage{authblk}
\usepackage{csquotes}
\usepackage[T1]{fontenc} 
\usepackage{lipsum}
\usepackage{float}

\usepackage{setspace}
\usepackage{titlesec}
\titleformat{\section} 
  {\normalfont\Large\bfseries}{\thesection.}{1em}{}

\theoremstyle{plain}
\newtheorem{theorem}{Theorem}
\newtheorem{lemma}[theorem]{Lemma}
\newtheorem{proposition}[theorem]{Proposition}
\newtheorem{corollary}[theorem]{Corollary}
\newtheorem{claim}{Claim}

\theoremstyle{definition}
\newtheorem{definition}{Definition}
\newtheorem{example}{Example}

\title{The Design of Optimally Balanced Pay-as-you-go Social Security Systems}

\author[1,2]{Leandro Lyra Braga Dognini}
\affil[1]{\small Department of Economics, Rio de Janeiro State University}
\affil[2]{\small Legislative Advisory, Federal Senate of Brazil} 
\date{\today} 

\begin{document}
\maketitle
\begin{abstract}
\noindent This paper bridges social security design and general equilibrium theory to conceive optimally balanced pay-as-you-go systems. The design is based on the backward calculation algorithm from \textcite{Dognini_2025}, which is used to find optimal monetary equilibria of prone-to-savings non-stationary overlapping generations economies with heterogeneous households. In particular, this algorithm makes the design applicable for reforming pay-as-you-go systems in countries undergoing demographic transitions. Due to households balanced budgets under equilibrium prices (i.e., Walras' law), these optimally balanced pay-as-you-go systems resemble the well-known notional accounts systems. The design is illustrated in a simplified framework using the past and forecast demographic and productivity dynamics of Brazil, China, India, Italy, and the United States from 1950 to 2070.
\end{abstract}

\textbf{Keywords}: Unfunded social security, pay-as-you-go, notional accounts, non-financial defined contribution, transition paths. 


\section{Introduction}\label{sec1}

\textsc{Demographic transitions} are a major challenge that countries are facing due to widespread plummeting fertility rates, and sustaining fiscal balance and economic growth with an elderly and shrinking population is no easy task. Converging life expectancy and fertility rates between historically different countries are leading to a similar future demographic landscape, although the speed of transition varies greatly between them.

For instance, Brazil and the United States will have the same ratio of working-age population to elderly population by 2054, according to the most recent World Population Prospects (see \hyperref[fig1]{Figure 1} below); although Brazil will face, in less than 40 years, the demographic transition that took over a century in the United States. China will go through the same demographic shift in about 25 years, and India, although still a younger country, is also heading in this direction\footnote{According to Figure \ref{fig1}, the dependence ratio goes from 8.2 to 2.6: between 2017 and 2052 in Brazil; between 2012 and 2038 in China; between 2028 and 2071 in India; between 1950 and 2024 in Italy; and between 1950 and 2052 in the United States.}.

\begin{figure}[h]
\includegraphics[width=11cm]{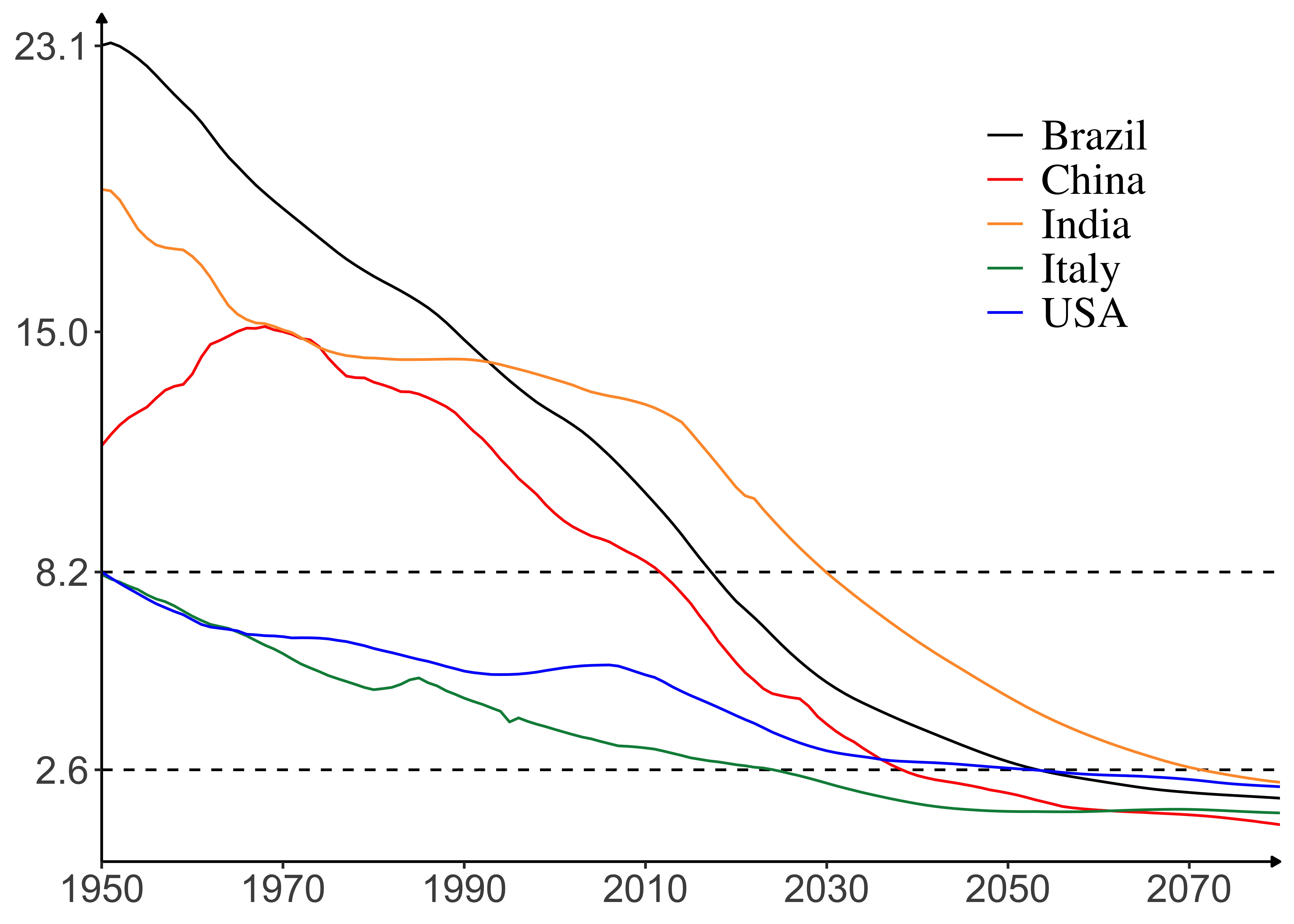}
\centering
\caption{Time evolution of the dependence ratio calculated as population with age between 15 and 64 years divided by population with age greater than 64 years \parencite{UN2024Pop}.}
\label{fig1}
\end{figure}

One cannot underestimate the repercussions of this transition. In 1950, Brazil had 23 working-age persons for each elderly one, and, less than a century later, it will not have 3. China's population may have already peaked and will not only get older but also smaller in the near future. Although countries in Europe have seen a more stable demography in the past decades, their fertility rates are currently reaching unprecedentedly low levels. It is not just a matter of cold projections; it is an entirely different global economic outlook in the blink of an eye.

A ubiquitous consequence of these demographic shifts is the upward pressure on pay-as-you-go social security systems. On one side, an elderly population increases the outflows from the existing systems; on the other, a contracting working-age population decreases (without substantial productivity gains) the inflows. Fiscal unbalances arise, leading to an increase in public debt, followed by a surge in taxes and reforms of pensions and other benefits (normally in this order due to political constraints).

For instance, in the United States, the 2024 OASDI Trustees Report \parencite{OASDI_2024} states that: persistent low fertility has led to a downward revision of total fertility rate assumptions; social security outflows are projected to be higher than inflows for all years from now on; and reserves are capable of covering costs only for the next 10 years\footnote{In June 2025, the Board of Trustees notified the President of the Senate and the Speaker of the House of Representatives: ``We are writing to notify you that we project that the reserves held in the Federal Old-Age and Survivors Insurance (OASI) Trust Fund will fall below 20 percent of annual cost by December 31, 2034, without legislation to address the imbalance between Social Security program revenues and benefits. The reserves expressed as a percentage of annual program cost of the OASI Trust Fund (the balance ratio) are projected to fall below 20 percent by the beginning of 2033. We project that the reserves of the OASI Trust Fund will be depleted soon after, during 2033, and only about 77 percent of benefits scheduled in current law will be payable at that time if no legislative action is taken.''}.

Given this scenario, this paper aims to build on the overlapping generations literature and lay a theoretical basis for the design of pay-as-you-go social security systems in a general demographic setup (thus including demographic transitions).

Following \textcite{Samuelson_1958}, overlapping generations models have become a well-established framework for the study of social security systems\footnote{Asymmetric information models are also a common means to study social security (e.g., \textcite{DiamondMirrlees_1978, DiamondMirrlees_1986}), though they are not addressed in this paper.}. For instance, \textcite{Barro_1974} uses them to demonstrate that bequests can offset inter-generational transfers; therefore, a pay-as-you-go social security system may have no real effect in this case (see \textcite{1993_Altig} for examples of when this offsetting effect does not hold). \textcite{Sheshinski_1981} adds lifetime uncertainty (e.g., \textcite{1965_Yaari}) to this model in order to derive optimal levels of social security benefits under a partial equilibrium analysis.

\textcite{Samuelson_1975} studies, in a stationary 2-period model with production (e.g., \textcite{Diamond_1965}), the optimal level of benefits paid and real-capital owned by a social security system that ensures a steady-state under the ``goldenest golden rule'' \parencite[p. 539]{Samuelson_1975} and \textcite{Feldstein_1985} uses this same sort of model to analyze the optimal level of social security benefits when individuals are myopic, i.e., ``lack the foresight to provide for their own old age'' \parencite[p. 303]{Feldstein_1985}. A comprehensive survey on these topics regarding social security design was written by \textcite{FeldsteinLiebman_2002}.

Furthermore, \textcite{Auerbach_1987} build a 55-period computable model with production and exogenously defined replacement rates to study the tradeoffs between unfunded and funded systems (e.g., Section C from Chapter 10 describes the crowding-out effect on capital\footnote{See \textcite{Kotlikoff_1979} for empirical findings on the effects of social security on savings.}) and the social security policy response to a demographic transition (e.g., Section F from Chapter 11 analyzes a cut from 60\% to 40\% in the replacement rate of the pay-as-you-go system \parencite[p. 175]{Auerbach_1987}).

Along this computable general equilibrium approach, \textcite{ImrohorogluJoines_1995,ImrohorogluJoines_1999} build a 65-period stochastic model to ``examine the optimality of and the welfare benefits associated with alternative social security arrangements'' \parencite[p. 84]{ImrohorogluJoines_1995}, concluding that a 30\% replacement rate is optimal for the conceived unfunded social security system. \textcite{Huggett_1999} deal with a 80-period model and compare its steady-states under different proposals for the US social security system.

More recently, \textcite{KruegerKubler_2002} worked with a stochastic overlapping generations model with incomplete markets to show that the introduction of a pay-as-you-go social security system can be Pareto improving. The same sort of conclusion is reached by \textcite{BallMankiw2007} and \textcite{GottardiKubler_2011}, where a pay-as-you-go social security system can lead to Pareto improvements due to inter-generational risk-sharing.

Moreover, \textcite{Conesa_2008} used overlapping generations models to calculate optimal transitions between unfunded and funded systems, and \textcite{Fehr_2013} deploy a 85-period stochastic model to find ``the optimal flat and earnings-related benefit combination that maximizes aggregate efficiency'' \parencite[p. 95]{Fehr_2013}.

With this literature overview in mind, the main contribution of this paper can now be clearly stated. The paper provides an algorithm for calculating the contribution and replacement rates of a pay-as-you-go social security system with a heterogeneous population of arbitrary lifespan (i.e., we are not restricted to models with a 2-period household lifespan) and a general demographic setup (i.e., we are not restricted to well-behaved demographic patterns, like stationary ones) that yields a fiscally balanced and Pareto optimal allocation of resources, thus resulting in an \textit{optimally balanced pay-as-you-go social security system}. 

This design is based on a consumption-loan deterministic overlapping generations model (therefore, the model does not deal with labor supply, production, or stochastic features) and the backward calculation algorithm from \textcite{Dognini_2025}, which has also been used to derive a solution for the Hahn problem \parencite{Hahn_1965, Bewley_1983} in monetary theory\footnote{Although applicable to monetary theory, the design of optimally balanced social security systems first motivated the results in \textcite{Dognini_2025}. As stated in its final paragraph: ``Finally, the specific economic problem that motivated the results of this paper should be clearly stated. Worldwide plummeting fertility rates are leading to a rapid demographic transition, putting high pressure on pay-as-you-go social security systems. The optimal design of these systems can be achieved through overlapping generations models and, specifically, through non-stationary consumption-loan models with heterogeneous households. However, this design requires the calculation of at least one efficient equilibrium.''}.

Furthermore, since the design is based on an overlapping generations model, Walras' law implies that we may regard the optimal contribution and replacement rates in terms of the balances of a fictitious account that keeps records of all contributions from and benefits paid to a given household. Therefore, this design leads to the well-known \textit{notional accounts} systems adopted, for instance, by Sweden and Italy, and a comprehensive survey on these systems was written by \textcite{HolzmannPalmer_2006}.

Finally, I highlight two fundamental aspects of these notional accounts systems that make them suitable for social security design: (i) there is an explicit link between contributions and benefits, thus mitigating labor market distortions (see the discussion on the ``efficiency gains from benefit-tax linkage'' in \textcite[pp. 155--161]{Auerbach_1987}); and (ii) parametric adjustments follow the demographic and productivity dynamics, thus placing a ``forward looking'' perspective on the system to keep it optimally balanced (see Section VI in \textcite[p. 17--21]{Diamond_2004} for a discussion on ``automatic and legislated adjustments'' of social security systems).

After this \hyperref[sec1]{Introduction}, the outline of the paper is as follows. \hyperref[sec2]{Section 2} defines the design of optimally balanced social security systems in a simple demographic setup and \hyperref[sec3]{Section 3} extends it to a general setup. \hyperref[sec4]{Section 4} then illustrates the design method in a simplified framework that builds on past and forecast demographic and productivity dynamic in Brazil, China, India, Italy, and the United States from 1950 to 2070. \hyperref[sec5]{Section 5} brings the concluding remarks. All proofs are stated in the \hyperref[appx]{Appendix}.

\section{Social security design in a simple demographic setup}\label{sec2}

Time is indexed by discrete periods $t\in\mathbb{N}$, there is a single perishable commodity in each period, and the government runs a pay-as-you-go social security system. The insured households are indexed by $h\in \mathbb{N}$ and are gathered in generations $G_{t}=\{h\in\mathbb{N}\mid h \textrm{ is born in period }t\}$, $t\geq0$, according to their period of birth and live for two periods, the one they are born and the next, except for those from generation $G_{0}$. Generation $G_{t}$ has a finite number of identical households given by $H_{t}\in\mathbb{N}$ and $\alpha_{t}=\alpha=H_{t+1}/H_{t}$, $t\geq0$, represents the constant demographic growth rate of the insured population.

Household $h\in G_{t}$, $t\geq1$, is defined through a consumption set $X^{h}=\mathbb{R}^{2}_{+}$, a common nonzero endowment $e^{h}=(e^{h}_{t},e^{h}_{t+1})=(e,0)\in\mathbb{R}^{2}_{+}$ and a common utility function $u^{h}:\mathbb{R}^{2}_{+}\rightarrow\mathbb{R}$, $u^{h}(c_{t},c_{t+1})=u(c_{t},c_{t+1})=\sqrt{c_{t}}+\sqrt{c_{t+1}}$. Household $h\in G_{0}$ is defined through a consumption set $X^{h}=\mathbb{R}_{+}$, a zero endowment $e^{h}=0\in\mathbb{R}_{+}$ and a common utility function $u^{h}:\mathbb{R}_{+}\rightarrow\mathbb{R}$, $u^{h}(c_{1})=\sqrt{c_{1}}$. 

The pay-as-you-go social security system is given by $\{s^{h}\}_{h\geq1}$, with $s^{h}=(s^{h}_{t},s^{h}_{t+1})\in\mathbb{R}^{2}$, $s^{h}+e^{h}\in\mathbb{R}^{2}_{+}$, $h\in G_{t}$, $t\geq1$, and $s^{h}=s^{h}_{1}\in\mathbb{R}$, $s^{h}+e^{h}=s^{h}\in\mathbb{R}_{+}$, $h\in G_{0}$. I proceed with the following definition.

\begin{definition}\label{defOptmBalanced}
    The social security system $\{s^{h}\}_{h\geq1}$ is \textit{balanced} if $\sum_{t\geq0}\sum_{h\in G_{t}} s^{h}=0$, and is \textit{optimally balanced} if, in addition, $u^{h}(e^{h}+s^{h})\geq u^{h}(e^{h})$, $h\geq1$, and the allocation $\{e^{h}+s^{h}\}_{h\geq1}$ is Pareto optimal.
\end{definition}

The minimum one can ask of a pay-as-you-go social security system is for it to be balanced. This means that contributions collected from the younger generations cover, but do not exceed, the benefits promised to the older generations (i.e., there are always enough resources to honor social security obligations, and the government sets contributions to the system only for this purpose). Furthermore, an optimally balanced system is one that makes no household worse off than it would be without it and has exhausted all unanimous possibilities for welfare increases (i.e., there is no possible Pareto improvement).

Notice that the government can design many balanced social security systems. For instance, let $\{s^{h}_{\sigma}\}_{h\geq1}$, $\sigma\in[0,1]$, be given by $s^{h}_{\sigma}=(-\sigma e,\alpha\sigma e)$, $t\geq1$, and $s^{h}_{\sigma}=s^{h}_{\sigma 1}= \alpha\sigma e$, $h\in G_{0}$. Traditionally, $\sigma\in[0,1]$ is called the \textit{contribution rate} and $\alpha\sigma>0$ the \textit{replacement rate}. Then,
\begin{eqnarray*}
    \sum_{t\geq0}\sum_{h\in G_{t}} s^{h}_{\sigma}=\sum_{t\geq1}\biggr(\sum_{h\in G_{t-1}} s^{h}_{\sigma t}+\sum_{h\in G_{t}} s^{h}_{\sigma t}\biggr)=\sum_{t\geq1}(H_{t-1}\alpha\sigma e-H_{t}\sigma e)=0,
\end{eqnarray*}
so that $\{s^{h}_{\sigma}\}_{h\geq1}$, $\sigma\in[0,1]$, is balanced. Since
\begin{eqnarray*}
    \frac{\alpha}{1+\alpha}=\textrm{argmax}_{\sigma\in[0,1]}\sqrt{(1-\sigma)e}+\sqrt{\alpha\sigma e},
\end{eqnarray*}
the system $\{s^{h}_{\alpha/(1+\alpha)}\}_{h\geq1}$ Pareto dominates $\{s^{h}_{\sigma}\}_{h\geq1}$, $\sigma\in[0,\alpha/(1+\alpha))$. Thus, the stationary system $\{s^{h}_{\sigma}\}_{h\geq1}$, $\sigma\in[0,\alpha/(1+\alpha))$, cannot be optimally balanced, and the government should not adopt it when designing the pay-as-you-go social security.

The obvious candidate at this point for designing our optimally balanced system is $\{s^{h}_{\alpha/(1+\alpha)}\}_{h\geq1}$. However, a closer look at Definition \ref{defOptmBalanced} reveals that we must compare $\{e^{h}+s^{h}_{\alpha/(1+\alpha)}\}_{h\geq1}$ to all possible allocations, not just to stationary ones. 

Therefore, up to this point, we do know that $\{s^{h}_{\sigma}\}_{h\geq1}$, $\sigma\in[0,\alpha/(1+\alpha))$ is not optimally balanced, but we do not know if $\{s^{h}_{\alpha/(1+\alpha)}\}_{h\geq1}$ is. To find this out, we need to redraft the insured population we have been working with as a fictitious competitive infinite-horizon economy.

Let $\mathcal{E}$ be an overlapping generations economy with fiat money composed of all the generations in our insured population, with household $h\in G_{0}$ holding a common amount of money $m^{h}=1$. Household $h\in G_{t}$, $t\geq1$, has a Walrasian demand function $x^{h}:\mathbb{R}^{2}_{++}\rightarrow \mathbb{R}^{2}_{+}$, $x^{h}(p_{t},p_{t+1})=(x^{h}_{t}(p_{t},p_{t+1}),x^{h}_{t+1}(p_{t},p_{t+1}))$, $(p_{t},p_{t+1})\in\mathbb{R}^{2}_{++}$, given by 
\begin{eqnarray*}
    x^{h}(p_{t},p_{t+1})=\textrm{argmax}_{c\in\mathbb{R}^{2}_{+}}& u^{h} (c) \\
    \textrm{ s.t. }& (p_{t},p_{t+1})\cdot (c-e^{h})\leq0,
\end{eqnarray*}
which can be simplified due to the common utility function and endowment to
\begin{eqnarray*}
    x^{h}(p_{t},p_{t+1})=x(p_{t},p_{t+1})=\biggr(\frac{p_{t+1}e}{p_{t}+p_{t+1}},\frac{p_{t}^{2}e}{p_{t+1}(p_{t}+p_{t+1})}\biggr).
\end{eqnarray*}
Also, the Walrasian demand of household $h\in G_{0}$, $x^{h}:\mathbb{R}^{2}_{++}\rightarrow \mathbb{R}_{+}$, is given by
\begin{eqnarray*}
    x^{h}(p_{m},p_{1})=\textrm{argmax}_{c\in\mathbb{R}_{+}}& u^{h} (c) \\
    \textrm{ s.t. }& p_{1}(c-e^{h})\leq p_{m}m^{h},
\end{eqnarray*}
which can be simplified due to the common utility function, endowment, and fiat money holdings to $x^{h}(p_{m},p_{1})=p_{m}/p_{1}$. A \textit{monetary equilibrium} of the economy $\mathcal{E}$ is a sequence $p=(p_{m},p_{1},\ldots)=(1,p_{1},\ldots)\in\mathbb{R}^{\infty}_{++}$ that satisfies market clearing in all periods, i.e.,
\begin{eqnarray}\label{eqMarketClearing}
   \sum_{t\geq0}\sum_{h\in G_{t}}x^{h}(p_{t},p_{t+1})-e^{h}=0,
\end{eqnarray}
with $p_{0}=p_{m}=1$. Notice that if we have a monetary equilibrium $p\in\mathbb{R}^{\infty}_{++}$ of our fictitious economy $\mathcal{E}$, then we can use it to design a social security system $\{s^{h}_{p}\}_{h\geq1}$ by letting $s^{h}_{p}=x^{h}(p_{t},p_{t+1})-e^{h}$, $h\in G_{t}$, $t\geq0$, and the market clearing equations (\ref{eqMarketClearing}) ensure that it is balanced.

It is convenient to define $\phi:\mathbb{R}_{++}\rightarrow\mathbb{R}_{++}$, $\phi(r)=r/(1+r)$. Let $r_{t}=p_{t}/p_{t+1}$, $t\geq0$, so that $\{r_{t}\}_{t\geq0}$ completely defines the sequence $p=(1,p_{1},p_{2},\ldots)\in\mathbb{R}^{\infty}_{++}$. Then, the market clearing equations (\ref{eqMarketClearing}) can be written as
\begin{eqnarray*}
    \alpha_{0}\phi(r_{1})&=&r_{0}\\
    \alpha_{1}\phi(r_{2})&=&r_{1}\phi(r_{1})\\
    \alpha_{2}\phi(r_{3})&=&r_{2}\phi(r_{2})\\
    &\vdots&
\end{eqnarray*}
with $\alpha_{t}=\alpha>0$, $t\geq0$. A careful analysis of this system of infinite equations (e.g., \textcite{Gale_1973}) reveals that $p=(1,p_{1},\ldots)\in\mathbb{R}^{\infty}_{++}$ is a monetary equilibrium if, and only if, $r_{1}\in(0,\alpha]$, $r_{0}=\alpha\phi(r_{1})$ and $r_{t+1}=\phi^{-1}(r_{t}\phi(r_{t})/\alpha)$, $t\geq1$. In particular, $r_{1}\in(0,\alpha]$ indexes all possible monetary equilibria.

It is well-known that the First Welfare Theorem does not hold in overlapping generations economies (e.g., \textcite{Samuelson_1958, CassYaari_1966, Shell_1971}) and so we cannot ensure that all allocations derived after $r_{1}\in (0,\alpha]$ are Pareto optimal. Nonetheless, as pointed out by \textcite[p. 286, Definition 2.5]{BalaskoShell_1980}, all such monetary equilibria define \textit{weakly Pareto optimal} allocations (i.e., there is no possible redistribution over a \textit{finite} number of periods that leads to a Pareto improvement). Therefore, we are half-way to an optimally balanced social security system.

Next, let $\psi:\mathbb{R}_{++}\rightarrow\mathbb{R}_{++}$ be given by
\begin{eqnarray*}
    \psi(r)=\frac{r+\sqrt{r^{2}+4r}}{2},
\end{eqnarray*}
so that $\psi(r\phi(r))=r$, $r>0$, and $f_{t}:\mathbb{R}_{++}\rightarrow\mathbb{R}_{++}$, $t\geq1$, be given by $f_{t}(r)=\psi(\alpha_{t}\phi(r))$. The next result follows the backward calculation algorithm from \textcite{Dognini_2025}\footnote{Although Proposition \ref{propOptReturnRates1Dim} is similar to Theorem 12 in \textcite{Dognini_2025}, we cannot simply use the latter since our fictitious economy does not satisfy one of its assumptions (namely, Assumption 3) due to $\sum_{h\in G_{t}}e^{h}_{t+1}=0$, $t\geq1$.}.

\begin{proposition}\label{propOptReturnRates1Dim}
Let $\mathcal{E}$ be the overlapping generations economy with fiat money above and $0<\alpha_{\min}\leq \alpha_{t}\leq \alpha_{\max}$, $t\geq0$. Then, there is a unique Pareto optimal monetary equilibrium $\tilde{p}=(1,\tilde{p}_{1},\ldots)\in\mathbb{R}^{\infty}_{++}$, with $\tilde{r}_{1}=\tilde{p}_{1}/\tilde{p}_{2}=\lim_{t\rightarrow\infty} f_{1}\circ \ldots \circ f_{t}(L)$, $L>0$. 
\end{proposition}

Proposition \ref{propOptReturnRates1Dim} furnishes a way of calculating the unique Pareto optimal monetary equilibrium of our fictitious economy $\mathcal{E}$ with a bounded (not necessarily stationary) demographic dynamic. The intuition behind the result is straightforward. Notice that the system of (infinite) market clearing equations can be ``approximated'' by the (finite) system
\begin{eqnarray*}
    \alpha_{0}\phi(r_{1})&=&r_{0}\\
    \alpha_{1}\phi(r_{2})&=&r_{1}\phi(r_{1})\\
    &\vdots&\\
    \alpha_{t}\phi(L)&=&r_{t}\phi(r_{t}),
\end{eqnarray*}
which can also be written, by applying $\psi(\cdot)$ from the second equation onward, as
\begin{eqnarray*}
    f_{0}(r_{1})&=&\psi(r_{0})\\
    f_{1}(r_{2})&=&r_{1}\\
    &\vdots&\\
    f_{t}(L)&=&r_{t}.
\end{eqnarray*}
Solving backward these equations yields $r_{1}=f_{1}\circ \ldots \circ f_{t}(L)$. Therefore, Proposition \ref{propOptReturnRates1Dim} states that this solution correctly approaches the unique Pareto optimal monetary equilibrium of $\mathcal{E}$ as $t\rightarrow\infty$, for all $L>0$.

Next, the fact that we are dealing with the stationary case (i.e., $\alpha_{t}=\alpha$, $t\geq0$) makes the calculation of the Pareto optimal monetary equilibrium trivial. Notice that $f_{t}(\alpha)=f(\alpha)=\alpha$, $t\geq1$, and $\tilde{r}_{1}=\lim_{t\rightarrow\infty} f_{1}\circ \ldots \circ f_{t}(\alpha)=\lim_{t\rightarrow\infty}f^{t}(\alpha)=\alpha$. Then, the unique Pareto optimal monetary equilibrium is 
\begin{eqnarray*}
\tilde{p}=\biggr(1,\frac{1+\alpha}{\alpha^{2}},\frac{1+\alpha}{\alpha^{3}},\ldots\biggr),    
\end{eqnarray*}
and the corresponding social security design is
\begin{eqnarray*}
    s^{h}_{\tilde{p}}=x^{h}(\tilde{p}_{t},\tilde{p}_{t+1})-e^{h}=\biggr(\frac{-\alpha e}{1+\alpha},\frac{\alpha^{2} e}{1+\alpha} \biggr)=s^{h}_{\alpha/(1+\alpha)}
\end{eqnarray*}
for $h\in G_{t}$, $t\geq1$. Clearly, $s^{h}_{\tilde{p}}=s^{h}_{\alpha/(1+\alpha)}$, $h\in G_{0}$. We conclude that our optimally balanced social security system is indeed $\{s^{h}_{\alpha/(1+\alpha)}\}_{h\geq1}$.

Proposition \ref{propOptReturnRates1Dim} also tells us how to reform the social security system $\{s^{h}_{\alpha/(1+\alpha)}\}_{h\geq1}$ after a marginal change in the demographic growth rate $\alpha_{t}$, $t\geq1$, due to the following corollary.

\begin{corollary}\label{corMarginalDemChange}
Let $\mathcal{E}$ be the overlapping generations economy with fiat money above and $\alpha_{t}=\alpha>0$, $t\geq0$. Then,
\begin{eqnarray*}
      \frac{\partial \tilde{r}_{1}}{\partial \alpha_{t}}=\frac{1+\alpha}{(2+\alpha)^{t}},
\end{eqnarray*}
for $t\geq1$.
\end{corollary}

Corollary \ref{corMarginalDemChange} states that changes in future growth rates \textit{reverberate} all the way down to the first period return rate, and the entire design of the corresponding social security system must be reformed in order to keep it optimally balanced. For instance,
\begin{eqnarray*}
    \frac{\partial x^{h}_{1}(\tilde{r}_{1},1)}{\partial \alpha_{t}}=-\frac{e}{(1+\alpha)^{2}}\frac{\partial \tilde{r}_{1}}{\partial \alpha_{t}}=-\frac{e}{(1+\alpha)(2+\alpha)^{t}},
\end{eqnarray*}
for $h\in G_{1}$, $t\geq1$. Therefore, an increase in future demographic growth rates leads to an increase in the contribution rates of generation $G_{1}$ and in the replacement rates of generation $G_{0}$. Similarly, a decrease in future demographic growth rates leads to a decrease in the contribution rates of generation $G_{1}$ and in the replacement rates of generation $G_{0}$. 

Although the discussions and results (i.e., Proposition \ref{propOptReturnRates1Dim} and Corollary \ref{corMarginalDemChange}) have been conducted so far in this simple two-period overlapping generations setup, all fundamental aspects of the design of optimally balanced pay-as-you-go social security systems can already be highlighted.

The first aspect is that optimally balanced systems are a good benchmark for our design goal. They not only ensure that no household is worse off due to the system and that contributions always cover benefits (and only cover them), but also that there is no other system that leads to a Pareto improvement (i.e., there are no unanimous welfare gains ``left on the table'').

Then, we come to our second fundamental aspect. In general, it is not easy to find optimally balanced systems (especially in non-stationary economies with heterogeneous households), as one can infer from the previous discussion and results. A way to achieve this is to conceive a fictitious overlapping generations economy $\mathcal{E}$ formed by the insured population and then look for a Pareto optimal equilibrium of it. 

An attentive look at Definition \ref{defOptmBalanced} reveals that this is a reasonable approach since: (i) utility maximization implies that no household is worse off due to the system; (ii) market clearing equations imply that the system is balanced and weakly Pareto optimal; and (iii) Pareto optimality implies that the system is optimally balanced. Briefly stated, we use this fictitious market mechanism to design an optimally balanced pay-as-you-go system.   

And how does one actually find this Pareto optimal equilibrium, in general? The answer is given by the \textcite{Cass_1972} criterion and the backward calculation algorithm in \textcite{Dognini_2025}. The key point is to work with well-behaved tail economies, and the next section aims precisely to define these. Once the tail economies are given, one ``moves'' backward, solving equilibrium equations ``step by step'' in order to approximate at least one Pareto optimal equilibrium of the fictitious economy $\mathcal{E}$. 

The third fundamental aspect is that the design of optimally balanced systems is forward looking (even though the calculation algorithm goes backward!). In a given sense, today's pensions must be defined by asking how much the younger generations are willing to contribute to the system in order to ensure fair pensions for themselves in the future. But this answer, in turn, depends on how much the next generations will be willing to contribute, and so forth in a never-ending pay-as-you-go system.

Therefore, if the demographic dynamic is altered (in the short or long run), the equilibrium equations imply a change in all pensions and contributions of the social security system from day one (see Corollary \ref{corMarginalDemChange}). We say that future demographic shifts \textit{reverberate} on today's social security rules because, for example, if future return rates of the system are lower than previously expected, current pensions and contributions must be adjusted in order to attain a fair distribution of the burden. Otherwise, it is just the younger generations paying more for each time less.

This forward-looking design is in sharp contrast to how the government and common sense view social security. The former is normally concerned with the short-run balance of the system and leaves all possible reform burdens to the next administration, since adjusting pensions and lowering contribution rates is hardly a preferential path (considering Corollary \ref{corMarginalDemChange} and a country with an ongoing demographic contraction). The latter usually views social security with a backward-looking eye: ``People have always retired at 55 with a formidable pension; of course, it's going to be the same for me!''

The fourth and final fundamental aspect is the following. Since this design of social security systems uses a fictitious market mechanism and, in particular, balanced budgets under common prices (i.e., Walras' law), it is natural to interpret the series of contributions and benefits of a given household as the balance of a \textit{notional account} that is adjusted over time according to the equilibrium \textit{return rates} of the fictitious economy $\mathcal{E}$. This is why these are called \textit{notional accounts social security systems}\footnote{The literature also adopts the term non-financial defined contribution (NDC) systems \parencite{BarrDiamond_2006,HolzmannPalmer_2006}. However, since optimality usually requires that both the contribution and replacement rates must be adjusted simultaneously, the term \textit{notional accounts systems} is preferable.}.

Furthermore, although our focus in this section and, more broadly, in this paper, is on demographic transitions, it is clear that the adjustment of these notional accounts is determined not only by the demographic dynamic but also by the productivity (i.e., endowment) dynamic and the preferences of each generation. Also, incorporating preferences into social security design is an important milestone of this paper.

Finally, the main shortfalls of this design should also be clearly stated. The first is that the demographic and productivity dynamics are exogenous and, therefore, we are assuming that they are not affected by the design of the social security system itself. Nonetheless, as far as we are dealing with incremental reforms and continue to assess the optimality of the system, this concern becomes less relevant. Second, the design depends heavily on demographic and productivity forecasts. The uncertainty of forecasts, therefore, will reflect on the predictability of the systems rules. One could think of some sort of ``buffer fund'' or ``compensation mechanism'' to share the burden of this uncertainty between generations, but I will not take on this matter in this paper.

The next section reveals how optimally balanced social security systems can be designed in a general demographic setup.

\section{Social security design in a general demographic setup}\label{sec3}

This section closely follows the definitions and results in \textcite{Dognini_2025}. Time is indexed by $t\in\mathbb{N}_{0}$, there is a single perishable commodity in each period, and the government runs a pay-as-you-go social security system. The insured households are indexed by $h\in\mathbb{N}$ and are gathered in generations $G_{t}=\{h\in\mathbb{N}\mid h\textrm{ is born in period }t\}$, $t\geq0$, and live for $S\geq2$ periods ($S$ stands for ``span of life''). Generation $G_{t}$ has a finite number of households given by $H_{t}\in\mathbb{N}$ and $\alpha_{t}=H_{t+1}/H_{t}$ represents the demographic growth rate.

Household $h\in G_{t}$ is defined by a consumption set $X^{h}=\mathbb{R}^{S}_{+}$, a nonzero endowment $e^{h}=(e^{h}_{t},\ldots,e^{h}_{t+S})\in\mathbb{R}^{S}_{+}$ (with $\sum_{t\geq0}\sum_{h\in G_{t}}e^{h}\in\mathbb{R}^{\infty}_{++}$) and a continuous, non-decreasing, semi-strictly quasiconcave\footnote{Utility function $u(\cdot)$ is \textit{semi-strictly quasiconcave} (see \textcite[p. 87, Definition 2]{ArrowHahn_1971}), if $u(c^{\prime})\geq u(c)$ implies $u(\alpha c^{\prime}+(1-\alpha) c)\geq u(c)$, $0\leq \alpha \leq 1$, and if $u(c^{\prime})>u(c)$ implies $u(\alpha c^{\prime}+(1-\alpha) c)>u(c)$, $0< \alpha \leq1$. In particular, if $u(\cdot)$ is semi-strictly quasiconcave, then it is quasiconcave.} utility function $u^{h}:\mathbb{R}^{S}_{+}\rightarrow\mathbb{R}$ without local maxima. Finally, all these households form a fictitious overlapping generations economy $\mathcal{E}$.

The first remark about $\mathcal{E}$ is that, unlike Section \ref{sec2}, all households live for the same number of periods $S\geq2$. This can be seen as a model of a \textit{long-standing} social security system and is a convenient way of making money holdings endogenous in our fictitious economy $\mathcal{E}$ (see \textcite{Dognini_2025} for a comprehensive discussion on this matter).

The second remark is that $S\geq2$ actually models the \textit{maximum lifespan} of households. This is because for $h\in G_{t}$, $t\geq0$, and $2\leq S^{\prime}<S$, we can always set $e^{h}_{t+s}=0$, $S^{\prime}<s\leq S$, and $u^{h}(c_{t},\ldots,c_{t+S})=u^{h}(c_{t},\ldots,c_{t+S^{\prime}})$ (i.e., $u^{h}(\cdot)$ does not depend on $c_{t+s}$, $S^{\prime}<s\leq S$), thereby representing a household with a shorter lifespan.

The third remark is that the economy $\mathcal{E}$ is equivalent to a two-period model $\tilde{\mathcal{E}}$ due to the relabeling algorithm conceived by \textcite[pp. 319-321]{BalaskoCassShell_1980}. This is illustrated in the following example.

\begin{example}\label{ex1}
Let $S=3$, so that $e^{h}=(e^{h}_{t},e^{h}_{t+1},e^{h}_{t+2})\in\mathbb{R}^{3}_{+}$ and $u^{h}:\mathbb{R}^{3}_{+}\rightarrow\mathbb{R}$, $u^{h}(c)=u^{h}(c_{t},c_{t+1},c_{t+2})$, $c\in\mathbb{R}^{3}_{+}$, $h\in G_{t}$, and $t\geq0$. Next, build the ``relabeled'' economy $\tilde{\mathcal{E}}$ by defining generation $\tilde{G}_{t}=G_{2t}\cup G_{2t+1}$, $t\geq0$. Then, two successive periods of $\mathcal{E}$ turn into a single period of $\tilde{\mathcal{E}}$, and the demographic growth rates become
    \begin{eqnarray}\label{eqGrowthRateBCS}
        \tilde{\alpha}_{t}=\frac{\tilde{H}_{t+1}}{\tilde{H}_{t}}=\frac{H_{2t+2}+H_{2t+3}}{H_{2t}+H_{2t+1}}=\frac{\alpha_{2t}\alpha_{2t+1}(1+\alpha_{2t+2})}{1+\alpha_{2t}},
    \end{eqnarray}
for $t\geq0$. Consequently, there are two perishable goods in each period of $\tilde{\mathcal{E}}$ and the endowments $\tilde{e}^{h}\in\mathbb{R}^{4}_{+}$, $h\geq1$, are given by
    \begin{eqnarray*}
        \tilde{e}^{h}=(\tilde{e}^{h}_{t},\tilde{e}^{h}_{t+1})=\begin{cases}
            (e^{h}_{2t},e^{h}_{2t+1},e^{h}_{2t+2},0),\textrm{ if } h\in G_{2t}\\
            (0,e^{h}_{2t+1},e^{h}_{2t+2},e^{h}_{2t+3}),\textrm{ if } h\in G_{2t+1}
        \end{cases}
    \end{eqnarray*}
for $h\in \tilde{G}_{t}$, $t\geq0$. Furthermore, utility function $\tilde{u}^{h}:\mathbb{R}^{4}_{+}\rightarrow\mathbb{R}$ is given by
    \begin{eqnarray*}
        \tilde{u}^{h}(\tilde{c}_{t},\tilde{c}_{t+1})=\begin{cases}
            u^{h}(\tilde{c}_{t1},\tilde{c}_{t2},\tilde{c}_{(t+1)1}),\textrm{ if } h\in G_{2t}\\
            u^{h}(\tilde{c}_{t2},\tilde{c}_{(t+1)1},\tilde{c}_{(t+1)2}),\textrm{ if } h\in G_{2t+1}
        \end{cases}
    \end{eqnarray*}
for $h\in \tilde{G}_{t}$, $t\geq0$. Notice that household $h\in \tilde{G}_{t}$ lives exactly two periods and, therefore, $\mathcal{E}$ is equivalent to the two-period model $\tilde{\mathcal{E}}$. Finally, I highlight an important convention that is adopted throughout the paper: all variables that stand with a $\sim$ on top (e.g., $\tilde{e}^{h}$, $\tilde{p}_{t}$ and $\tilde{x}^{h}(\tilde{p}_{t},\tilde{p}_{t+1})$) relate to the two-period ``relabeled'' model $\tilde{\mathcal{E}}$, whereas the ones without it relate to the original economy $\mathcal{E}$. 
\end{example}

Following the convention established in Example \ref{ex1}, let $\tilde{\mathcal{E}}$ be the two-period model equivalent to $\mathcal{E}$ through the \textcite[pp. 319-321]{BalaskoCassShell_1980} relabeling algorithm, with $\tilde{L}_{t}\geq1$ perishable goods in each period $t\geq0$. Each household $h\in \tilde{G}_{t}$, $t\geq0$, has a Walrasian demand function $\tilde{x}^{h}:\mathbb{R}^{\tilde{L}_{t}+\tilde{L}_{t+1}}_{++}\rightarrow \mathbb{R}^{\tilde{L}_{t}+\tilde{L}_{t+1}}_{+}$, $\tilde{x}^{h}(\tilde{p}_{t},\tilde{p}_{t+1})=(\tilde{x}^{h}_{t}(\tilde{p}_{t},\tilde{p}_{t+1}),\tilde{x}^{h}_{t+1}(\tilde{p}_{t},\tilde{p}_{t+1}))$, $(\tilde{p}_{t},\tilde{p}_{t+1})\in\mathbb{R}^{\tilde{L}_{t}+\tilde{L}_{t+1}}_{++}$, given by 
\begin{eqnarray*}
    \tilde{x}^{h}(\tilde{p}_{t},\tilde{p}_{t+1})=\textrm{argmax}_{\tilde{c}\in\mathbb{R}^{\tilde{L}_{t}+\tilde{L}_{t+1}}_{+}}& \tilde{u}^{h} (\tilde{c}) \\
    \textrm{ s.t. }& (\tilde{p}_{t},\tilde{p}_{t+1})\cdot (\tilde{c}-\tilde{e}^{h})\leq0.
\end{eqnarray*}
The excess demand function in period $t\geq1$, $\tilde{z}_{t}:\mathbb{R}^{\tilde{L}_{t-1}+\tilde{L}_{t}+\tilde{L}_{t+1}}_{++}\rightarrow \mathbb{R}^{\tilde{L}_{t}}$, is 
\begin{eqnarray*}
    \tilde{z}_{t}(\tilde{p}_{t-1},\tilde{p}_{t},\tilde{p}_{t+1})&=&\sum_{h\in \tilde{G}_{t-1}}\tilde{z}^{h}_{t}(\tilde{p}_{t-1},\tilde{p}_{t})+\sum_{h\in \tilde{G}_{t}}\tilde{z}^{h}_{t}(\tilde{p}_{t},\tilde{p}_{t+1})\\
    &=&\sum_{h\in \tilde{G}_{t-1}}(\tilde{x}^{h}_{t}(\tilde{p}_{t-1},\tilde{p}_{t})-\tilde{e}^{h}_{t})+\sum_{h\in \tilde{G}_{t}}(\tilde{x}^{h}_{t}(\tilde{p}_{t},\tilde{p}_{t+1})-\tilde{e}^{h}_{t}).
\end{eqnarray*}
Also, the \textit{real savings} $\tilde{s}^{h}:\mathbb{R}^{\tilde{L}_{t}+\tilde{L}_{t+1}}_{++}\rightarrow \mathbb{R}$ of household $h\in \tilde{G}_{t}$, $t\geq0$, is
    \begin{eqnarray*}
       \tilde{s}^{h}(\tilde{p}_{t},\tilde{p}_{t+1}) = \frac{\tilde{p}_{t}}{\Vert \tilde{p}_{t}\Vert}\cdot (\tilde{e}^{h}_{t}-\tilde{x}^{h}_{t}(\tilde{p}_{t},\tilde{p}_{t+1})).
    \end{eqnarray*}

All assumptions in \textcite{Dognini_2025} are also made here for the economy $\tilde{\mathcal{E}}$ and rewritten below for convenience. 

Therefore, I assume there are $0<\alpha_{\min}<\alpha_{\max}$ and $e_{\max}>0$ such that $\tilde{\alpha}_{t}\in[\alpha_{\min},\alpha_{\max}]$, $t\geq0$, and $\Vert \tilde{e}^{h}\Vert_{\infty}=\max_{1\leq i\leq  \tilde{L}_{t}+\tilde{L}_{t+1}}\vert \tilde{e}^{h}_{i}\vert\leq e_{\max}$, $h\in \mathbb{N}$, and, for $j\geq 0$, all households $h\in \bigcup^{j}_{t=0}\tilde{G}_{t}$ are indirectly resource related\footnote{Following \textcite[117-118]{ArrowHahn_1971}, given $j\geq0$, household $h^{\prime}\in\bigcup^{j}_{t=0}\tilde{G}_{t}$ is \textit{resource related} to household $h^{\prime\prime}\in\bigcup^{j}_{t=0}\tilde{G}_{t}$ if, for every allocation $\{c^{h}\}_{h\in\bigcup^{j}_{t=0}\tilde{G}_{t}}$, $c^{h}\in\mathbb{R}^{\tilde{L}_{t}+\tilde{L}_{t+1}}_{+}$, $h\in \tilde{G}_{t}$, $0\leq t\leq  j$, with $\sum_{0\leq t\leq  j}\sum_{h\in \tilde{G}_{t}}c^{h}\leq \sum_{0\leq t\leq  j}\sum_{h\in \tilde{G}_{t}}\tilde{e}^{h}$, there exists an allocation $\{\bar{c}^{\, h}\}_{h\in\bigcup^{j}_{t=0}\tilde{G}_{t}}$ and $y\in\mathbb{R}^{\sum^{j+1}_{i=0}\tilde{L}_{i}}_{+}$ such that $\sum_{0\leq t\leq  j}\sum_{h\in \tilde{G}_{t}}\bar{c}^{\, h}\leq\sum_{0\leq t\leq  j}\sum_{h\in \tilde{G}_{t}}\tilde{e}^{h} +y$,
\begin{eqnarray*}
        u^{h}(\bar{c}^{\, h})&\geq&u^{h}(c^{h}),\textrm{ for all }h\in\bigcup^{j}_{t=0}\tilde{G}_{t},\\
        u^{h^{\prime\prime}}(\bar{c}^{\, h^{\prime\prime}})&>&u^{h^{\prime\prime}}(c^{h^{\prime\prime}}),
\end{eqnarray*}
and
\begin{eqnarray*}
    y_{ti}>0\textrm{ only if } \tilde{e}^{h^{\prime}}_{ti}>0,
\end{eqnarray*}
for $i\in\{1,\ldots, \tilde{L}_{t}\}$, $0\leq t\leq j$. Household $h^{\prime}$ is \textit{indirectly resource related} to household $h^{\prime\prime}$ if there is a sequence of households $\{h_{m}\}_{0\leq m\leq n}$, $n\geq1$, with $h_{0}=h^{\prime}$, $h_{n}=h^{\prime\prime}$ and $h_{m}$ resource related to $h_{m+1}$, $0\leq m\leq n-1$. 
}. 

Before stating the next assumption, let $\beta=\max\, \{1+\alpha^{-1}_{\min},1+\alpha_{\max}\}$ and $\mathcal{B}_{t}(\sigma)=\{\tilde{p}\in\mathbb{R}^{\tilde{L}_{t}+\tilde{L}_{t+1}}_{++}\mid \sigma \leq \tilde{p}_{i}\leq \sigma^{-1}, 1\leq i \leq \tilde{L}_{t}+\tilde{L}_{t+1}\}$, for $\sigma\in(0,1)$, $t\geq 0$. Then, for all $t\geq0$, there is $\sigma_{t}\in(0,1)$ such that
\begin{eqnarray*}
    (\tilde{p}_{t}/\tilde{p}_{t1},\tilde{p}_{t+1}/\tilde{p}_{t1})\notin \mathcal{B}_{t}(\sigma_{t})\implies \biggr\Vert \sum_{h\in \tilde{G}_{t}} \frac{\tilde{x}^{h}(\tilde{p}_{t},\tilde{p}_{t+1})}{\tilde{H}_{t}}\biggr\Vert_{\infty} > \beta e_{\max},
\end{eqnarray*}
for $(\tilde{p}_{t},\tilde{p}_{t+1})\in \mathbb{R}^{\tilde{L}_{t}+\tilde{L}_{t+1}}_{++}$. The \textit{set of equilibria} $\tilde{\mathcal{H}}\subset\mathbb{R}^{\infty}_{++}$ of economy $\tilde{\mathcal{E}}$ is
\begin{eqnarray*}
\tilde{\mathcal{H}}=\{\tilde{p}\in\mathbb{R}^{\infty}_{++}\mid \tilde{p}_{01}=1,(\tilde{p}_{0},\tilde{p}_{1})\in\mathcal{B}_{0}(\sigma_{0})\textrm{ and } \tilde{z}_{t}(\tilde{p}_{t-1},\tilde{p}_{t},\tilde{p}_{t+1})=0, t\geq 1\},
\end{eqnarray*} 
and the \textit{subset of Pareto optimal equilibria} $\mathcal{H}^{PO}\subseteq\mathcal{H}$ is
\begin{eqnarray*}
    \tilde{\mathcal{H}}^{PO}=\{\tilde{p}\in\tilde{\mathcal{H}}\mid \nexists \{y^{h}\}_{h\geq1}, y^{h}\in\mathbb{R}^{\tilde{L}_{t}+\tilde{L}_{t+1}}_{+}, h\in \tilde{G}_{t},t\geq0, \textrm{ s.t. }  \sum_{t\geq0}\sum_{h\in \tilde{G}_{t}}(y^{h}-\tilde{x}^{h}(\tilde{p}_{t},\tilde{p}_{t+1}))=0\\
    \textrm{and }u^{h}(y^{h})\geq u^{h}(\tilde{x}^{h}(\tilde{p}_{t},\tilde{p}_{t+1})), h\in \tilde{G}_{t}, t\geq0, \textrm{with at least one strict inequality}\}.
\end{eqnarray*}
Assume that the \textcite{Cass_1972} criterion can be applied to $\tilde{\mathcal{E}}$ so that, for $\tilde{p}\in\tilde{\mathcal{H}}$, 
\begin{eqnarray*}
    \tilde{p}\in\tilde{\mathcal{H}}^{PO}\iff \sum^{+\infty}_{t=0}\frac{1}{H_{t}\Vert \tilde{p}_{t}\Vert}=+\infty.
\end{eqnarray*}
The last assumption is that the economy $\tilde{\mathcal{E}}$ is \textit{prone to savings}, meaning that there are $\varepsilon>0$, $\delta>0$ such that, for $t\geq0$, $(\tilde{p}_{t},\tilde{p}_{t+1})\in\mathbb{R}^{\tilde{L}_{t}+\tilde{L}_{t+1}}_{++}$ and $(\tilde{p}_{t}/\tilde{p}_{t1},\tilde{p}_{t+1}/\tilde{p}_{t1})\in\mathcal{B}_{t}(\sigma_{t})$,
\begin{eqnarray*}
    \sum_{h\in \tilde{G}_{t}} \frac{\tilde{s}^{h}(\tilde{p}_{t},\tilde{p}_{t+1})}{\tilde{H}_{t}}\leq \delta\implies \frac{\Vert \tilde{p}_{t} \Vert}{\Vert \tilde{p}_{t+1} \Vert} \leq \frac{\tilde{\alpha}_{t}}{1+\varepsilon}.
\end{eqnarray*}
Next, let $\{\tilde{\mathcal{E}}_{k}\}_{k\geq1}$ be a sequence of two-period overlapping generations economies that satisfy all assumptions above. Then, $\{\tilde{\mathcal{E}}_{k}\}_{k\geq1}$ is \textit{finitely replicating} $\tilde{\mathcal{E}}$ if: (i) there are common values of $\alpha_{\min},\alpha_{\max}$ and $e_{\max}>0$; (ii) there exists a compact set $\mathcal{K}^{\prime}\subset\mathbb{R}^{\infty}$ such that $\tilde{\mathcal{H}}\subseteq\mathcal{K}^{\prime}$ and $\tilde{\mathcal{H}}_{k}\subseteq\mathcal{K}^{\prime}$, $k\geq1$; and (iii) all generations of $\tilde{\mathcal{E}}_{k}$ born until period $k\geq1$ are identical to those of $\tilde{\mathcal{E}}$. In this case, the economy $\tilde{\mathcal{T}}_{k}$ formed by all generations of $\tilde{\mathcal{E}}_{k}$ born after period $k\geq1$ is called a \textit{tail economy}. Finally, as a topological space, $\mathbb{R}^{\infty}$ is endowed with the product topology.

With all definitions and assumptions made, I can now state the theorem from \textcite{Dognini_2025} that we use to design optimally balanced social security systems.

\begin{theorem}[Backward calculation algorithm]\label{theoAlgorithmFinal}
If $\{\tilde{\mathcal{E}}_{k}\}_{k\geq1}$ is finitely replicating $\tilde{\mathcal{E}}$, then 
\begin{eqnarray*}
\lim_{j\rightarrow\infty}\overline{\bigcup_{k\geq j}\tilde{\mathcal{H}}^{PO}_{k}}\subseteq \tilde{\mathcal{H}}^{PO},
\end{eqnarray*}
with $\lim_{j\rightarrow\infty}\overline{\bigcup_{k\geq j}\tilde{\mathcal{H}}^{PO}_{k}}$ a non-empty compact set.
\end{theorem}

Theorem \ref{theoAlgorithmFinal} states that there is at least one Pareto optimal equilibrium of $\tilde{\mathcal{H}}^{PO}$ that can be arbitrarily approximated (in the product topology) by elements of $\bigcup_{k\geq j}\tilde{\mathcal{H}}^{PO}_{k}$, $j\geq1$. The key point of this result is the finitely replicating sequence $\{\tilde{\mathcal{E}}_{k}\}_{k\geq1}$ and, therefore, the sequence of tail economies $\{\tilde{\mathcal{T}}_{k}\}_{k\geq1}$, since the latter completely defines the former. 

In order to find elements of $\tilde{\mathcal{H}}^{PO}_{k}$ (and, ultimately, of $\tilde{\mathcal{H}}^{PO}$) we must choose well-behaved tail-economies $\tilde{\mathcal{T}}_{k}$, $k\geq1$, in the sense that we know precisely their set of Pareto optimal equilibria, so that we can solve equilibrium equations backward, step-by-step, as in Proposition \ref{propOptReturnRates1Dim} (after some careful thought, one can notice that the set of Pareto optimal equilibria of the tail economies fulfills the role of the ``boundary condition'' $L>0$ in Proposition \ref{propOptReturnRates1Dim}).

There are two natural options for building these well-behaved tail economies: minimum-lifespan and full-lifespan tail economies. The following example illustrates the minimum-lifespan case.

\begin{example}[Minimum-lifespan tail economies]\label{ex2}
I build on Example \ref{ex1}. Our goal is to build a well-behaved replicating sequence $\{\tilde{\mathcal{E}}_{k}\}_{k\geq1}$, which is equivalent to properly choosing tail economies $\{\tilde{\mathcal{T}}_{k}\}_{k\geq1}$. Let $\tilde{L}^{\tau k}_{t}\geq1$ (the symbol $\tau$ stands for ``tail'') be the number of perishable goods in period $t\geq0$ of $\tilde{\mathcal{T}}_{k}$, $k\geq1$. Since $S=3$, we have $\tilde{L}_{t}=2$, $t\geq0$ (i.e., economy $\tilde{\mathcal{E}}$ has two perishable goods in every period). Since all generations of $\tilde{\mathcal{E}}_{k}$ coincide with those of $\tilde{\mathcal{E}}$ until period $k\geq1$, we have $\tilde{L}^{\tau k}_{0}=\tilde{L}_{k+1}=2$. 
    
Next, we must choose $\tilde{L}^{\tau k}_{t}$, $t\geq1$. Since we are dealing with the minimum-lifespan case, let $\tilde{L}^{\tau k}_{t}=1$, $t\geq1$.

Then, we can take, for instance, generation $\tilde{G}^{\tau k}_{0}$ composed of $\tilde{\mathcal{H}}_{k}$ identical households (so that the population remains constant from period $k$ to period $k+1$ in $\tilde{\mathcal{E}}_{k}$) with a common endowment $\tilde{e}^{\tau k0}=(\tilde{e}^{\tau k0}_{01},\tilde{e}^{\tau k0}_{02},\tilde{e}^{\tau k0}_{1})=(e_{y},e_{y},e_{o})\in\mathbb{R}^{3}_{++}$, $0<e_{o}<e_{y}$ (the indexes stand for ``young'' and ``old''), and a common utility function $\tilde{u}^{\tau k 0}:\mathbb{R}^{3}_{++}\rightarrow\mathbb{R}$ given by
\begin{eqnarray*}
    \tilde{u}^{\tau k0}(\tilde{c}_{0},\tilde{c}_{1})=\tilde{u}^{\tau k0}(\tilde{c}_{01},\tilde{c}_{02},\tilde{c}_{1})=\theta \log \tilde{c}_{01}+\theta \log \tilde{c}_{02}+(1-2\theta)\log \tilde{c}_{1}.
\end{eqnarray*}
for $\theta\in(0,1/2)$. Furthermore, generation $\tilde{G}^{\tau k}_{t}$, $t\geq1$, is composed of $\tilde{\mathcal{H}}_{k}$ households (so that the population is constant in the tail economy $\tilde{\mathcal{T}}_{k}$), with a common endowment $\tilde{e}^{\tau kt}=(e_{y},e_{o})\in\mathbb{R}^{2}_{++}$, and a common utility function $\tilde{u}^{\tau k t}:\mathbb{R}^{2}_{+}\rightarrow\mathbb{R}$, given by $\tilde{u}^{\tau k t}(\tilde{c}_{t},\tilde{c}_{t+1})=\log \tilde{c}_{t}+\log \tilde{c}_{t+1}$. After some careful thought on these tail economies (see the analysis in \textcite{Gale_1973}), we conclude that
\begin{eqnarray}\label{eqParetoOptEqSetEx2}
    \tilde{\mathcal{H}}^{PO}_{\tau k}=\{\tilde{p}\in\mathbb{R}^{\infty}_{++}\mid \tilde{p}_{01}=1, (\tilde{p}_{0},\tilde{p}_{1})\in \mathcal{B}(\sigma_{0}),(1-2\theta)(\tilde{p}_{0},\tilde{p}_{1})\cdot \tilde{e}^{\tau k0}/\tilde{p}_{1}=(e_{y}+e_{o})/2,\nonumber\\
    \tilde{p}_{t}=\tilde{p}_{t+1}, t\geq1\},
\end{eqnarray}
with
\begin{eqnarray*}
    \frac{(1-2\theta)(\tilde{p}_{0},\tilde{p}_{1})\cdot \tilde{e}^{\tau k0}}{\tilde{p}_{1}}=\frac{e_{y}+e_{o}}{2}\iff \tilde{p}_{1}=\frac{2(1-2\theta)\Vert \tilde{p}_{0}\Vert e_{y}}{e_{y}+(4\theta-1)e_{o}}=f(\tilde{p}_{0}).
\end{eqnarray*}
Once the tail economy $\tilde{\mathcal{T}}_{k}$, $k\geq1$, is defined, to calculate $\mathcal{H}^{PO}_{k}$ we must solve
\begin{eqnarray*}
    \sum_{h\in G_{0}}z^{h}_{2}(p_{0},p_{1},p_{2})+\sum_{h\in G_{1}}z^{h}_{2}(p_{1},p_{2},p_{3})+\sum_{h\in G_{2}}z^{h}_{2}(p_{2},p_{3},p_{4})&=&0\\
    &\vdots&\\
    \sum_{h\in G_{2k}}z^{h}_{2k+2}(p_{2k},\ldots)+\sum_{h\in G_{2k+1}}z^{h}_{2k+2}(p_{2k+1},\ldots)+\sum_{h\in \tilde{G}^{\tau k}_{0}}\tilde{z}^{h}_{01}(p_{2k+2},p_{2k+3},p_{2k+4})&=&0\\
    \sum_{h\in G_{2k+1}}z^{h}_{2k+3}(p_{2k+1},p_{2k+2},p_{2k+3})+\sum_{h\in \tilde{G}^{\tau k}_{0}}\tilde{z}^{h}_{02}(p_{2k+2},p_{2k+3},p_{2k+4})&=&0,
\end{eqnarray*}
with $f(p_{2k+2},p_{2k+3})=p_{2k+4}$, so that $(p_{2k+2}/p_{2k+2},p_{2k+3}/p_{2k+2},p_{2k+4}/p_{2k+2})\in\tilde{\mathcal{H}}^{PO}_{\tau k}$. Notice, first, that this \textit{finite} system is analogous to the finite system in Section \ref{sec2} that leaded to Proposition \ref{propOptReturnRates1Dim}. Here, again, our goal is to solve backward the $2k+2$ equations, but we are no longer ``departing'' from $L>0$ as our ``boundary condition''. This time, our boundary is formed by prices $(p_{2k+2},p_{2k+3},p_{2k+4})\in\mathbb{R}^{3}_{++}$ that satisfy $f(p_{2k+2},p_{2k+3})=p_{2k+4}$. 

To ease the calculation, we can normalize $p_{2k+4}=2(1-2\theta)e_{y}/(e_{y}+(4\theta-1)e_{o})=\beta$, so that $f(p_{2k+2},p_{2k+3})=p_{2k+4}$ implies $\Vert(p_{2k+2},p_{2k+3})\Vert=p_{2k+2}+p_{2k+3}=1$. Let $p_{2k+2}=\sigma\in(0,1)$, so that the previous system can now be written as
\begin{eqnarray*}
    \sum_{h\in G_{0}}z^{h}_{2}(p_{0},p_{1},p_{2})+\sum_{h\in G_{1}}z^{h}_{2}(p_{1},p_{2},p_{3})+\sum_{h\in G_{2}}z^{h}_{2}(p_{2},p_{3},p_{4})&=&0\\
    &\vdots&\\
    \sum_{h\in G_{2k}}z^{h}_{2k+2}(p_{2k},p_{2k+1},\sigma)+\sum_{h\in G_{2k+1}}z^{h}_{2k+2}(p_{2k+1},\sigma, 1-\sigma)+\sum_{h\in \tilde{G}^{\tau k}_{0}}\tilde{z}^{h}_{01}(\sigma,1-\sigma,\beta)&=&0\\
    \sum_{h\in G_{2k+1}}z^{h}_{2k+3}(p_{2k+1},\sigma,1-\sigma)+\sum_{h\in \tilde{G}^{\tau k}_{0}}\tilde{z}^{h}_{02}(\sigma,1-\sigma,\beta)&=&0.
\end{eqnarray*}

Since $\sigma\in(0,1)$ is fixed, we have $2k+2$ equations and $2k+2$ unknown prices $(p_{0},\ldots,p_{2k+1})$. We can use the last equation to calculate $p_{2k+1}>0$; then, the previous to last to calculate $p_{2k}>0$; and we keep moving backward until we reach the first equation and calculate $p_{0}>0$. This algorithm will not necessarily reach its end (i.e., calculate $p_{0}>0$) for all $\sigma\in(0,1)$, but it will for at least one such value (see Theorem 8 in \textcite{Dognini_2025}). 

Furthermore, Theorem \ref{theoAlgorithmFinal} tells us that some of these solutions correctly approach a Pareto optimal monetary equilibrium of our fictitious economy $\mathcal{E}$, which is what we are looking for in order to design our optimally balanced social security system. This point will be further discussed in Section \ref{sec4}.
\end{example}

Next, we deal with the full-lifespan case.

\begin{example}[Full-lifespan tail economies]\label{ex3}
Once again, I build on Example \ref{ex1}. Our goal is to choose well-behaved tail economies $\{\tilde{\mathcal{T}}_{k}\}_{k\geq1}$. Since $S=3$, we have $\tilde{L}_{t}=2$, $t\geq0$, and $\tilde{L}^{\tau k}_{0}=\tilde{L}_{k+1}=2$, $k\geq1$. Next, we must choose $\tilde{L}^{\tau k}_{t}$, $t\geq1$. Since we are dealing with the full-lifespan case, let our tail economies be three-period overlapping generations ones, so that $\tilde{L}^{\tau k}_{t}=2$, $t\geq1$.

Also, let generation $G^{\tau k}_{0}$ be composed of $H^{\tau k}_{0}=H_{2(k+1)}$ households and $\mathcal{T}^{k}$ have a constant population growth rate $\alpha\geq1$, i.e., $H^{\tau k}_{t+1}/H^{\tau k}_{t}=\alpha$, $t\geq0$. Notice that $\tilde{G}^{\tau k}_{t}=G^{\tau k}_{2t}\cup G^{\tau k}_{2t+1}$, $t\geq0$, and, reasoning as in (\ref{eqGrowthRateBCS}), we conclude that $\tilde{\alpha}^{\tau k}_{t}=\alpha^{2}$, $t\geq0$.

Furthermore, all households hold a common endowment given by $e^{h}=(e,e,e)\in\mathbb{R}^{3}_{++}$, $h\geq1$, so that
    \begin{eqnarray*}
        \tilde{e}^{h}=(\tilde{e}^{h}_{t},\tilde{e}^{h}_{t+1})=\begin{cases}
            (e,e,e,0),\textrm{ if } h\in G^{\tau k}_{2t}\\
            (0,e,e,e),\textrm{ if } h\in G^{\tau k}_{2t+1}
        \end{cases}
    \end{eqnarray*}
for $h\in \tilde{G}^{\tau k}_{t}$, $t\geq0$. Households also have a common utility function $u^{h}:\mathbb{R}^{3}_{+}\rightarrow\mathbb{R}$,
\begin{eqnarray*}
    u^{h}(c_{t},c_{t+1},c_{t+2})=\log c_{t}+\theta\log c_{t+1}+\theta^{2}\log c_{t+2},
\end{eqnarray*}
for $h\in G^{\tau k}_{t}$, $t\geq0$, and $\theta\geq1$, so that
    \begin{eqnarray*}
        \tilde{u}^{h}(\tilde{c}_{t},\tilde{c}_{t+1})=\begin{cases}
             \log \tilde{c}_{t1} +  \theta\log \tilde{c}_{t2} + \theta^{2}\log \tilde{c}_{(t+1)1},\textrm{ if } h\in G^{\tau k}_{2t}\\
             \log \tilde{c}_{t2} + \theta\log \tilde{c}_{(t+1)1} + \theta^{2}\log \tilde{c}_{(t+1)2},\textrm{ if } h\in G^{\tau k}_{2t+1}
        \end{cases}
    \end{eqnarray*}
for $h\in \tilde{G}^{\tau k}_{t}$, $t\geq0$. Then, the Walrasian demand $\tilde{x}^{h}:\mathbb{R}^{4}_{++}\rightarrow\mathbb{R}^{4}_{+}$, $h\geq1$, is given by 
    \begin{eqnarray*}
        \tilde{x}^{h}(\tilde{p}_{t},\tilde{p}_{t+1})=\begin{dcases}
            \frac{(\tilde{p}_{t},\tilde{p}_{t+1})\cdot \tilde{e}^{h}}{1+\theta+\theta^{2}}\biggr(\frac{1}{\tilde{p}_{t1}},\frac{\theta}{\tilde{p}_{t2}},\frac{\theta^{2}}{\tilde{p}_{(t+1)1}},0\biggr),\textrm{ if } h\in G^{\tau k}_{2t}\\
           \frac{(\tilde{p}_{t},\tilde{p}_{t+1})\cdot \tilde{e}^{h}}{1+\theta+\theta^{2}}\biggr(0,\frac{1}{\tilde{p}_{t2}},\frac{\theta}{\tilde{p}_{(t+1)1}},\frac{\theta^{2}}{\tilde{p}_{(t+1)2}}\biggr),\textrm{ if } h\in G^{\tau k}_{2t+1}
        \end{dcases}
    \end{eqnarray*}
for $h\in\tilde{G}^{\tau k}_{t}$, $t\geq0$. Also, real savings per capita is given by
    \begin{eqnarray}\label{eqRealSavingsTail}
        \sum_{h\in\tilde{G}^{\tau k}_{t}}\frac{\tilde{s}^{h}(\tilde{p}_{t},\tilde{p}_{t+1})}{\tilde{H}^{\tau k}_{t}}&=&\sum_{h\in G^{\tau k}_{2t}}\frac{\tilde{p}_{t}\cdot(\tilde{e}^{h}_{t}-\tilde{x}^{h}_{t}(\tilde{p}_{t},\tilde{p}_{(t+1)}))}{\Vert \tilde{p}_{t}\Vert (H^{\tau k}_{2t}+H^{\tau k}_{2t+1})}+\sum_{h\in G^{\tau k}_{2t+1}}\frac{\tilde{p}_{t}\cdot(\tilde{e}^{h}_{t}-\tilde{x}^{h}_{t}(\tilde{p}_{t},\tilde{p}_{(t+1)}))}{\Vert \tilde{p}_{t}\Vert(H^{\tau k}_{2t}+H^{\tau k}_{2t+1})}\nonumber\\
        &=&\frac{\theta^{2}-((1+\theta)\tilde{p}_{(t+1)1}-\alpha(\theta+\theta^{2})\tilde{p}_{t2})/\Vert \tilde{p}_{t}\Vert-\alpha \Vert \tilde{p}_{t+1}\Vert/\Vert \tilde{p}_{t}\Vert}{(1+\theta+\theta^{2})(1+\alpha)}e,
    \end{eqnarray}
for $t\geq0$. I proceed with the following two lemmas.

\begin{lemma}\label{lemmaTailProneToSavings}
    Let $\tilde{\mathcal{T}}_{k}$, $k\geq1$, be a tail economy described above. If $1+\theta+\alpha<\theta^{2}\alpha^{2}$, then  $\tilde{\mathcal{T}}_{k}$ is prone to savings.
\end{lemma}

\begin{proposition}\label{propTailSetofEquilibria}
    Let $\tilde{\mathcal{T}}_{k}$, $k\geq1$, be a tail economy described above, with $\min\{\alpha,\theta\}>1$, and $\tilde{\mathcal{H}}_{\tau k}$ its set of equilibria. Also, let 
    \begin{eqnarray*}
        \lambda_{3}=\frac{-(1+\alpha)(1+\theta)+\sqrt{(1+\alpha)^{2}(1+\theta)^{2}-4\theta\alpha}}{2\alpha}.
    \end{eqnarray*}
    Then, $\tilde{p}\in \tilde{\mathcal{H}}_{\tau k}$ if, and only if, $(\tilde{p}_{0},\tilde{p}_{1})\in\mathcal{B}_{0}(\sigma_{0})$ and there are $a_{1}\geq0$, $a_{2}, a_{3}\in\mathbb{R}$, such that $a_{1}+a_{2}+a_{3}=1$ and the sequence given by
    \begin{eqnarray*}
        \begin{bmatrix}
            p_{t} \\ p_{t+1} \\ p_{t+2} \\ p_{t+3}
        \end{bmatrix} = 
        a_{1} \theta^{t} \begin{bmatrix}
            1 \\ \theta \\ \theta^{2} \\ \theta^{3}
        \end{bmatrix}+
         a_{2} \alpha^{-t} \begin{bmatrix}
            1 \\ 1/\alpha \\ 1/\alpha^{2} \\ 1/\alpha^{3}
        \end{bmatrix}+
         a_{3} \lambda^{t}_{3} \begin{bmatrix}
            1 \\ \lambda_{3} \\ \lambda_{3}^{2} \\ \lambda^{3}_{3}
        \end{bmatrix},
    \end{eqnarray*}
for $t\geq0$, is strictly positive.
\end{proposition}

Lemma \ref{lemmaTailProneToSavings} allows us to use Proposition 7 from \textcite{Dognini_2025} to state that, for $\tilde{p}\in\tilde{\mathcal{H}}_{\tau k}$, $\tilde{p}\notin \tilde{\mathcal{H}}^{PO}_{\tau k}$ if, and only if, 
\begin{eqnarray*}
    \lim_{t\rightarrow\infty}\sum_{h\in\tilde{G}^{\tau k}_{t}}\frac{\tilde{s}^{h}(\tilde{p}_{t},\tilde{p}_{t+1})}{\tilde{H}^{\tau k}_{t}}=0.
\end{eqnarray*}
Notice then that if $a_{1}>0$ in Proposition \ref{propTailSetofEquilibria}, we have 
\begin{eqnarray*}
    \lim_{t\rightarrow\infty}\biggr(\frac{p_{t}}{p_{t}},\frac{p_{t+1}}{p_{t}},\frac{p_{t+2}}{p_{t}},\frac{p_{t+3}}{p_{t}}\biggr)=(1,\theta,\theta^{2},\theta^{3}).
\end{eqnarray*}
Therefore, stationarity of the tail economy $\tilde{\mathcal{T}_{k}}$, continuity and homogeneity of degree zero of real savings, and (\ref{eqRealSavingsTail}) imply
\begin{eqnarray*}
     \lim_{t\rightarrow\infty}\sum_{h\in\tilde{G}^{\tau k}_{t}}\frac{\tilde{s}^{h}(\tilde{p}_{t},\tilde{p}_{t+1})}{\tilde{H}^{\tau k}_{t}}&=&\lim_{t\rightarrow\infty}\sum_{h\in\tilde{G}^{\tau k}_{t}}\frac{\tilde{s}^{h}(p_{2t},p_{2t+1},p_{2t+2},p_{2t+3})}{\tilde{H}^{\tau k}_{t}}\\
     &=&\sum_{h\in\tilde{G}^{\tau k}_{0}}\frac{\tilde{s}^{h}(1,\theta,\theta^{2},\theta^{3})}{\tilde{H}^{\tau k}_{0}}\\
     &=&\frac{\theta^{2}(1+\theta)-((1+\theta)\theta^{2}-\alpha(\theta+\theta^{2})\theta)-\alpha (\theta^{2}+\theta^{3})}{(1+\theta+\theta^{2})(1+\alpha)(1+\theta)}e\\
     &=&0.
\end{eqnarray*}
If $a_{1}=0$ in Proposition \ref{propTailSetofEquilibria}, notice that we must have $a_{2}>0$. Since $\Vert \lambda_{3}\Vert<1/\alpha$, then 
\begin{eqnarray*}
    \lim_{t\rightarrow\infty}\biggr(\frac{p_{t}}{p_{t}},\frac{p_{t+1}}{p_{t}},\frac{p_{t+2}}{p_{t}},\frac{p_{t+3}}{p_{t}}\biggr)=\biggr(1,\frac{1}{\alpha},\frac{1}{\alpha^{2}},\frac{1}{\alpha^{3}}\biggr).
\end{eqnarray*}
By the same reasoning, we have
\begin{eqnarray*}
     \lim_{t\rightarrow\infty}\sum_{h\in\tilde{G}^{\tau k}_{t}}\frac{\tilde{s}^{h}(\tilde{p}_{t},\tilde{p}_{t+1})}{\tilde{H}^{\tau k}_{t}}
     &=&\sum_{h\in\tilde{G}^{\tau k}_{0}}\frac{\tilde{s}^{h}(1,\alpha^{-1},\alpha^{-2},\alpha^{-3})}{\tilde{H}^{\tau k}_{0}}\\
     &=&\frac{(\theta-\alpha^{-1})(2(\theta+\alpha^{-1})+\theta\alpha^{-1}+1)}{(1+\theta+\theta^{2})(2+\alpha+\alpha^{-1})}e\\
     &>&0.
\end{eqnarray*}
This leads us to the following corollary.
\begin{corollary}\label{corHPO}
 Let $\tilde{\mathcal{T}}_{k}$, $k\geq1$, be a tail economy described above, with $1+\theta+\alpha<\theta^{2}\alpha^{2}$ and $\min\{\alpha,\theta\}>1$, and $\tilde{\mathcal{H}}^{PO}_{\tau k}$ its subset of Pareto optimal equilibria. Also, let 
    \begin{eqnarray*}
        \lambda_{3}=\frac{-(1+\alpha)(1+\theta)+\sqrt{(1+\alpha)^{2}(1+\theta)^{2}-4\theta\alpha}}{2\alpha}.
    \end{eqnarray*}
    Then, $\tilde{p}\in \tilde{\mathcal{H}}^{PO}_{\tau k}$ if, and only if, $(\tilde{p}_{0},\tilde{p}_{1})\in\mathcal{B}_{0}(\sigma_{0})$ and there is $a_{3}\in(-1,1/\alpha \vert \lambda_{3}\vert)$ such that
    \begin{eqnarray*}
        \begin{bmatrix}
            p_{t} \\ p_{t+1} \\ p_{t+2} \\ p_{t+3}
        \end{bmatrix} = 
         \frac{\alpha^{-t}}{1+a_{3}} \begin{bmatrix}
            1 \\ 1/\alpha \\ 1/\alpha^{2} \\ 1/\alpha^{3}
        \end{bmatrix}+
          \frac{a_{3}\lambda^{t}_{3}}{1+a_{3}} \begin{bmatrix}
            1 \\ \lambda_{3} \\ \lambda_{3}^{2} \\ \lambda^{3}_{3}
        \end{bmatrix},
    \end{eqnarray*}
for $t\geq0$.
\end{corollary}

Once the tail economy $\tilde{\mathcal{T}}_{k}$, $k\geq1$, is defined and its set of Pareto optimal equilibria is given by Corollary \ref{corHPO}, to calculate $\mathcal{H}^{PO}_{k}$ we must solve
\begin{eqnarray*}
    \sum_{h\in G_{0}}z^{h}_{2}(p_{0},p_{1},p_{2})+\sum_{h\in G_{1}}z^{h}_{2}(p_{1},p_{2},p_{3})+\sum_{h\in G_{2}}z^{h}_{2}(p_{2},p_{3},p_{4})&=&0\\
    &\vdots&\\
    \sum_{h\in G_{2k}}z^{h}_{2k+2}(p_{2k},\ldots)+\sum_{h\in G_{2k+1}}z^{h}_{2k+2}(p_{2k+1},\ldots)+\sum_{h\in G^{\tau k}_{0}}z^{h}_{0}(p_{2k+2},p_{2k+3},p_{2k+4})&=&0\\
    \sum_{h\in G_{2k+1}}z^{h}_{2k+3}(p_{2k+1},\ldots)+\sum_{h\in G^{\tau k}_{0}}z^{h}_{1}(p_{2k+2},p_{2k+3},p_{2k+4})+\sum_{h\in G^{\tau k}_{1}}\tilde{z}^{h}_{1}(p_{2k+3},p_{2k+4},p_{2k+5})&=&0,
\end{eqnarray*}
with $(p_{2k+2},p_{2k+3},p_{2k+4},p_{2k+5})=(1,\alpha^{-1},\alpha^{-2},\alpha^{-3})/(1+a_{3})+a_{3}(1,\lambda_{3},\lambda_{3}^{2},\lambda_{3}^{3})/(1+a_{3})$, $a_{3}\in (-1,1/\alpha\vert\lambda_{3}\vert)$, so that $(p_{2k+2},p_{2k+3},p_{2k+4},p_{2k+5})\in\tilde{\mathcal{H}}^{PO}_{\tau k}$ due to Corollary \ref{corHPO}.
\end{example}

A closer look at Examples \ref{ex2} and \ref{ex3} reveals that the terminal boundary conditions are one-dimensional. So why bother working with the full-lifespan case? Because when designing optimally balanced social security systems, the larger dimension of the boundary conditions implies that we generally have more options to choose from. Imagine you have an ongoing system, and in some sense, you would like to choose the optimally balanced system that is closer to it. Notice that, with this in mind, equilibrium multiplicity is not a problem but actually a desirable feature of the fictitious overlapping generations economy. 

And why did we not get this larger dimension in Example \ref{ex3}? Briefly, because we have chosen a tail economy that is too symmetric. We can elaborate further with the following intuitive argument. 

Notice that an economy with a lifespan $S\geq2$ is equivalent to a two-period model with $S-1$ goods in each period. As Example \ref{ex2} reveals, the minimum-lifespan tail economies have $S-1$ goods in $t=0$ and a single good afterward. Therefore, if the prices $(p_{0},p_{1})\in\mathbb{R}^{S}_{++}$ of an equilibrium are known, one can use the one-dimensional market clearing equation in $t=1$ to calculate $p_{2}>0$; then, the one-dimensional market clearing equation in $t=2$ to calculate $p_{3}>0$ and so forth. Since we normalize $p_{01}=1$, we conclude that the set of equilibria of these minimum-lifespan tail economies has a maximum dimension of $S-1$.

On the other hand, as Example \ref{ex3} reveals, the full-lifespan tail economies have $S-1$ goods in all periods $t\geq0$. Therefore, if the prices $(p_{0},p_{1})\in\mathbb{R}^{2S-2}_{++}$ of an equilibrium are known, one can use the $S-1$-dimensional market clearing equation in $t=1$ to calculate $p_{2}\in\mathbb{R}^{S-1}_{++}$; then, the $S-1$-dimensional market clearing equation in $t=2$ to calculate $p_{3}\in\mathbb{R}^{S-1}_{++}$ and so forth. Since we normalize $p_{01}=1$, we conclude that the set of equilibria of these full-dimension tail economies has a maximum dimension of $2S-3$. For $S=3$, this leads to a maximum dimension of $3$. However, Proposition \ref{propTailSetofEquilibria} reveals that the set of equilibria of our full-lifespan tail economy in Example \ref{ex3} has dimension $2<3$, and this leads to the equal dimension of the sets of Pareto optimal equilibria that furnish the terminal boundary conditions.

\section{Demographic transition of Brazil, China, India, Italy, and the United States}\label{sec4}
In this section, we apply the backward calculation algorithm to a three-period overlapping generations economy that adopts the demographic and productivity dynamics of Brazil, China, India, Italy, and the United States from 1950 to 2070. Therefore, the lifespan is $S=3$ with each period representing two decades (the first two periods represent ages between 20 and 59, and the last represents ages between 60 and 79), so that $t=0$ starts in 1950 and $t=5$ ends in 2070.

Generation $G_{t}$, $0\leq t\leq 5$, is composed of $H_{t}$ households, which is determined by the average number of households aged between 20 and 39 in the respective time window according to the latest World Population Prospects \parencite{UN2024Pop}. For example, $G_{0}$ is composed of the average number of households with ages between 20 and 39 in the 20-year time window from 1950 to 1969.

Household $h\in G_{t}$, $0\leq t\leq 5$, holds a common endowment $e^{h}=(e_{t},e_{t},\phi e_{t})\in\mathbb{R}^{3}_{+}$, $\phi\in[0,1]$, and has a common utility function 
$u^{h}:\mathbb{R}^{3}_{+}\rightarrow\mathbb{R}$,
\begin{eqnarray*}
    u^{h}(c_{t},c_{t+1},c_{t+2})=u(c_{t},c_{t+1},c_{t+2})= \log c_{t}+ \theta \log c_{t+1}+ \theta^{2}\log c_{t+2},
\end{eqnarray*}
for $\theta\geq1$. The endowment $e_{t}$ is given by the GDP per capita in the middle-year of the first working-life period according to \textcite{WorldBank2026GDP}. For example, $e_{0}, e_{3}>0$ are given by the GDP per capita in 1960 and 2020, respectively. The table below gathers the values of $\{H_{t}\}_{0\leq t\leq 5}$ and $\{e_{t}\}_{0\leq t\leq 3}$ for each country.

\begin{table}[H]
\centering
\begin{tabular}{ccccccccccc}
    & \multicolumn{2}{c}{Brazil} & \multicolumn{2}{c}{China} & \multicolumn{2}{c}{India} & \multicolumn{2}{c}{Italy} & \multicolumn{2}{c}{United States} \\ \cline{2-11} 
    & $H_{t}$ & $e_{t}$ & $H_{t}$ & $e_{t}$ & $H_{t}$ & $e_{t}$ & $H_{t}$ & $e_{t}$ & $H_{t}$ & $e_{t}$ \\ \hline
t=0 & 20 & 2603 & 180 & 241 & 122 & 313 & 15 & 9731 & 48 & 18899\\ \hline
t=1 & 36 & 6472 & 291 & 436 & 193 & 394 & 16 & 21901 & 71 & 30930\\ \hline
t=2 & 56 & 6817 & 439 & 2237 & 320 & 756 & 17 & 32609 & 83 & 48379\\ \hline
t=3 & 65 & 9435 & 406 & 10573 & 455 & 1806 & 14 & 29.469 & 91 & 59194\\ \hline
t=4 & 57 & - & 318 & - & 484 & - & 12 & - & 97 & - \\ \hline
t=5 & 47 & - & 190 & - & 427 & - & 9 & - & 95 & - \\ \hline
\end{tabular}
\caption{GDP per capita in 2015 US\$ \parencite{WorldBank2026GDP} is represented by $e_{t}$, $0\leq t\leq 3$. $H_{t}$, $0\leq t\leq 5$, is given by the two-decade average number (in millions) of people with ages between 20 and 39 \parencite{UN2024Pop}}
\label{table1}
\end{table}

Let $\beta_{t}=e_{t+1}/e_{t}$, $0\leq t\leq 2$, so that the (finite) sequences $\{\alpha_{t}\}_{0\leq t\leq 5}$ and $\{\beta_{t}\}_{0\leq t\leq 2}$ can be seen as representing the demographic and productivity dynamics of the economy, respectively. The table below gathers these values.

\begin{table}[H]
\centering
\begin{tabular}{ccccccccccc}
    & \multicolumn{2}{c}{Brazil} & \multicolumn{2}{c}{China} & \multicolumn{2}{c}{India} & \multicolumn{2}{c}{Italy} & \multicolumn{2}{c}{United States} \\ \cline{2-11} 
    & $\alpha_{t}$ & $\beta_{t}$ & $\alpha_{t}$ & $\beta_{t}$ & $\alpha_{t}$ & $\beta_{t}$ & $\alpha_{t}$ & $\beta_{t}$ & $\alpha_{t}$ & $\beta_{t}$ \\ \hline
t=0 & 1.78 & 2.49 & 1.62 & 1.81 & 1.58 & 1.26 & 1.07 & 2.25 & 1.49 & 1.64\\ \hline
t=1 & 1.58 & 1.05 & 1.51 & 5.13 & 1.66 & 1.92 & 1.09 & 1.49 & 1.17 & 1.56\\ \hline
t=2 & 1.16 & 1.38 & 0.92 & 4.73 & 1.42 & 2.39 & 0.8 & 0.9 & 1.10 & 1.22\\ \hline
t=3 & 0.87 & - & 0.78 & - & 1.06 & - & 0.87 & - & 1.06 & -\\ \hline
t=4 & 0.82 & - & 0.60 & - & 0.88 & - & 0.79 & - & 0.98 & - \\ \hline
\end{tabular}
\caption{The sequences $\{\alpha_{t}\}_{0\leq t\leq 4}$ and $\{\beta_{t}\}_{0\leq t\leq 2}$.}
\label{table2}
\end{table}
 
Notice that our productivity data in Table \ref{table1} only goes until $t=3$, with $e_{3}>0$ representing the 2020 GDP per capita. Since we need to define the endowments $e_{t}>0$, $4\leq t\leq 5$, we may do so by letting $\beta_{t}=\beta_{2}$, $3\leq t\leq 4$ (i.e., we forecast GDP per capita simply by maintaining the last growth rate from Table \ref{table2}). With this definition, it is convenient to let $\gamma_{t}=\alpha_{t}\beta_{t}$, $0\leq t\leq 4$, and the next table gathers these values.

\begin{table}[H]
\centering
\begin{tabular}{cccccc}
    & Brazil       & China        & India        & Italy        & United States \\ \cline{2-6} 
    & $\gamma_{t}$ & $\gamma_{t}$ & $\gamma_{t}$ & $\gamma_{t}$ & $\gamma_{t}$  \\ \hline
t=0 & 4.42 & 2.93 & 1.99 & 2.41 & 2.44 \\ \hline
t=1 & 1.67 & 7.75 & 3.18 & 1.62 & 1.83 \\ \hline
t=2 & 1.61 & 4.37 & 3.39 & 0.72 & 1.34 \\ \hline
t=3 & 1.21 & 3.71 & 2.54 & 0.78 & 1.29 \\ \hline
t=4 & 1.14 & 2.82 & 2.10 & 0.71 & 1.20 \\ \hline
\end{tabular}
\caption{The sequence $\{\gamma_{t}\}_{0\leq t\leq 4}$.}
\label{table3}
\end{table}

Our goal is to choose a convenient tail economy starting at $t=6$ (this is equivalent to period $t=3$ under \textcite{BalaskoCassShell_1980} relabeling) and then work with $\tilde{\mathcal{E}}_{2}$. Notice that, in fact, we do not know the behavior of all generations from $\tilde{\mathcal{E}}$, since we have only actually modeled the first three. 

The rationale we adopt is that by using a well-chosen $\tilde{\mathcal{E}}_{2}$ we can reasonably approximate a Pareto optimal equilibrium of $\tilde{\mathcal{E}}$ due to Theorem \ref{theoAlgorithmFinal}. Let $\tilde{\mathcal{E}}_{2}$ be defined through the full-lifespan tail economy from Example \ref{ex3}, with values of $H^{\tau}_{0}$, $e_{\tau}$, $\alpha_{\tau}$, and $\theta_{\tau}$ later specified. Then, equilibrium equations in periods $2\leq t\leq 7$ can be written as
\begin{eqnarray*}
    \frac{(p_{0}+p_{1}+\phi p_{2})\theta^{2}}{(1+\theta+\theta^{2})p_{2}} + \gamma_{0}\frac{(p_{1}+p_{2}+\phi p_{3})\theta}{(1+\theta+\theta^{2})p_{2}}+\gamma_{0}\gamma_{1}\frac{(p_{2}+p_{3}+\phi p_{4})}{(1+\theta+\theta^{2})p_{2}}&=&\phi +\gamma_{0}+\gamma_{0}\gamma_{1}\\
    &\vdots&\\
    \frac{(p_{3}+p_{4}+\phi p_{5})\theta^{2}}{(1+\theta+\theta^{2})p_{5}} + \gamma_{3}\frac{(p_{4}+p_{5}+\phi p_{6})\theta}{(1+\theta+\theta^{2})p_{5}}+\gamma_{3}\gamma_{4}\frac{(p_{5}+p_{6}+\phi p_{7})}{(1+\theta+\theta^{2})p_{5}}&=&\phi+\gamma_{3}+\gamma_{3}\gamma_{4}\\
    \frac{(p_{4}+p_{5}+\phi p_{6})\theta^{2}}{(1+\theta+\theta^{2})p_{6}} + \gamma_{4}\frac{(p_{5}+p_{6}+\phi p_{7})\theta}{(1+\theta+\theta^{2})p_{6}}+\frac{H^{\tau}_{0}e_{\tau}}{H_{4}e_{4}}\frac{(p_{6}+p_{7}+p_{8})}{(1+\theta_{\tau}+\theta_{\tau}^{2})p_{6}}&=&\phi+\gamma_{4}+\frac{H^{\tau}_{0}e_{\tau}}{H_{4}e_{4}}\\
    \frac{(p_{5}+p_{6}+\phi p_{7})\theta^{2}}{(1+\theta+\theta^{2})p_{7}} +\frac{H^{\tau}_{0}e_{\tau}}{H_{5}e_{5}}\frac{(p_{6}+p_{7}+p_{8})\theta_{\tau}}{(1+\theta_{\tau}+\theta_{\tau}^{2})p_{7}}+\frac{H^{\tau}_{1}e_{\tau}}{H_{5}e_{5}}\frac{(p_{7}+p_{8}+p_{9})}{(1+\theta_{\tau}+\theta_{\tau}^{2})p_{7}}&=&\phi+\frac{(1+\alpha_{\tau})H^{\tau}_{0}e_{\tau}}{H_{5}e_{5}}.
\end{eqnarray*}
Let $H^{\tau}_{0}=\alpha_{4}H_{5}$ and $e_{\tau}=\beta_{4}e_{5}$ be such that the demographic and productivity growth rates remain unchanged from period $t=4$ to period $t=5$ (i.e., $\alpha_{5}=\alpha_{4}$ and $\beta_{5}=\beta_{4}$). Then, equilibrium equations become
\begin{equation}\label{eqSystemEqEquations}
    \begin{aligned}
    \frac{(p_{0}+p_{1}+\phi p_{2})\theta^{2}}{(1+\theta+\theta^{2})p_{2}} + \gamma_{0}\frac{(p_{1}+p_{2}+\phi p_{3})\theta}{(1+\theta+\theta^{2})p_{2}}+\gamma_{0}\gamma_{1}\frac{(p_{2}+p_{3}+\phi p_{4})}{(1+\theta+\theta^{2})p_{2}}&=&\phi +\gamma_{0}+\gamma_{0}\gamma_{1}\\
    &\vdots&\\
    \frac{(p_{3}+p_{4}+\phi p_{5})\theta^{2}}{(1+\theta+\theta^{2})p_{5}} + \gamma_{3}\frac{(p_{4}+p_{5}+\phi p_{6})\theta}{(1+\theta+\theta^{2})p_{5}}+\gamma_{3}\gamma_{4}\frac{(p_{5}+p_{6}+\phi p_{7})}{(1+\theta+\theta^{2})p_{5}}&=&\phi+\gamma_{3}+\gamma_{3}\gamma_{4}\\
    \frac{(p_{4}+p_{5}+\phi p_{6})\theta^{2}}{(1+\theta+\theta^{2})p_{6}} + \gamma_{4}\frac{(p_{5}+p_{6}+\phi p_{7})\theta}{(1+\theta+\theta^{2})p_{6}}+\gamma_{4}^{2}\frac{(p_{6}+p_{7}+p_{8})}{(1+\theta_{\tau}+\theta_{\tau}^{2})p_{6}}&=&\phi+\gamma_{4}+\gamma_{4}^{2}\\
    \frac{(p_{5}+p_{6}+\phi p_{7})\theta^{2}}{(1+\theta+\theta^{2})p_{7}} +\gamma_{4}\frac{(p_{6}+p_{7}+p_{8})\theta_{\tau}}{(1+\theta_{\tau}+\theta_{\tau}^{2})p_{7}}+\gamma_{4}\alpha_{\tau}\frac{(p_{7}+p_{8}+p_{9})}{(1+\theta_{\tau}+\theta_{\tau}^{2})p_{7}}&=&\phi+(1+\alpha_{\tau})\gamma_{4}.
    \end{aligned}
\end{equation}

Corollary \ref{corHPO} implies that the equilibrium prices are calculated departing from the ``boundary prices'' given by 
    \begin{eqnarray*}
        \begin{bmatrix}
            p_{6} \\ p_{7} \\ p_{8} \\ p_{9}
        \end{bmatrix} = 
         \frac{1}{1+a_{3}}\begin{bmatrix}
            1 \\ 1/\alpha_{\tau}\\ 1/\alpha^{2}_{\tau} \\ 1/\alpha^{3}_{\tau}
        \end{bmatrix}+
         \frac{a_{3}}{1+a_{3}} \begin{bmatrix}
            1 \\ \lambda_{3} \\ \lambda_{3}^{2} \\ \lambda^{3}_{3}
        \end{bmatrix},
    \end{eqnarray*}
with $a_{3}\in(-1,1/\alpha_{\tau} \vert \lambda_{3}\vert)$ and
    \begin{eqnarray*}
        \lambda_{3}=\frac{-(1+\alpha_{\tau})(1+\theta_{\tau})+\sqrt{(1+\alpha_{\tau})^{2}(1+\theta_{\tau})^{2}-4\alpha_{\tau}\theta_{\tau}}}{2\alpha_{\tau}}.
    \end{eqnarray*}
Lemma \ref{lemmaTailProneToSavings} implies that by choosing $\theta_{\tau}\geq1$ such that $1+\theta_{\tau}+\alpha_{\tau}<\theta_{\tau}^{2}\alpha_{\tau}^{2}$, we obtain a prone-to-savings tail economy. The next result is analogous to Lemma \ref{lemmaTailProneToSavings} (actually, it is a generalization of it) and indicates values of $\theta\geq1$ that lead to a prone-to-savings economy $\tilde{\mathcal{E}_{2}}$.
\begin{lemma}\label{LemmaBrazilChinaUSPronetoSavings}
    Let $\tilde{\mathcal{E}}_{2}$ be the economy described above. If $\gamma_{2t}+(1+\theta)\phi<\tilde{\alpha}_{t}\theta^{2}$, $0\leq t\leq 2$, then there are $\varepsilon,\delta>0$ such that
    \begin{eqnarray*}
    \sum_{h\in\tilde{G}_{t}}\frac{\tilde{s}^{h}(\tilde{p}_{t},\tilde{p}_{t+1})}{\tilde{H}_{t}}\leq\delta
    \implies \frac{\Vert \tilde{p}_{t}\Vert}{\Vert \tilde{p}_{t+1}\Vert}\leq \frac{\tilde{\alpha}_{t}}{1+\varepsilon},
\end{eqnarray*}
for $(\tilde{p}_{t},\tilde{p}_{t+1})\in \mathbb{R}^{L_{t}+L_{t+1}}_{++}$, $0\leq t\leq 2$. 
\end{lemma}
Let $l:\mathbb{R}^{3}_{++}\rightarrow\mathbb{R}_{++}$ be given by
\begin{eqnarray*}
    l(\phi,\tilde{\alpha}_{t},\gamma_{2t})=\frac{\phi+\sqrt{\phi^{2}+4\tilde{\alpha}_{t}(\gamma_{2t}+\phi)}}{2\tilde{\alpha}_{t}},
\end{eqnarray*}
and notice that, for $\theta>0$, we have
\begin{eqnarray}\label{eqIff}
    \gamma_{2t}+(1+\theta)\phi<\tilde{\alpha}_{t}\theta^{2}\iff \theta>l(\phi,\tilde{\alpha}_{t},\gamma_{2t}),
\end{eqnarray}
for $0\leq t\leq 2$. Also, Following Example \ref{ex1}, we have
\begin{eqnarray*}
\tilde{\alpha}_{t}
=\frac{\alpha_{2t}\alpha_{2t+1}(1+\alpha_{2t+2})}{1+\alpha_{2t}},
\end{eqnarray*}
for $t=0,1,$ and
\begin{eqnarray*}
\tilde{\alpha}_{2}
=\frac{\alpha_{4}^{2}(1+\alpha_{\tau})}{1+\alpha_{4}}
\end{eqnarray*}
Let $\alpha_{\tau}=\gamma_{4}$, so that $\tilde{\alpha}_{2}=\alpha_{4}^{2}$. The following table gathers the lower bounds on $\theta$ from Lemma \ref{LemmaBrazilChinaUSPronetoSavings} and (\ref{eqIff}) for $\phi=0.2$.
\begin{table}[H]
\centering
\begin{tabular}{cccccc}
    & Brazil       & China        & India        & Italy        & United States \\ \cline{2-6}
    & $l(\phi,\tilde{\alpha}_{t},\gamma_{2t})$ & $l(\phi,\tilde{\alpha}_{t},\gamma_{2t})$ & $l(\phi,\tilde{\alpha}_{t},\gamma_{2t})$ & $l(\phi,\tilde{\alpha}_{t},\gamma_{2t})$ & $l(\phi,\tilde{\alpha}_{t},\gamma_{2t})$  \\ \hline
t=0 & 1.50 & 1.28 & 0.90 & 1.54 & 1.29 \\ \hline
t=1 & 1.57 & 2.69 & 1.70 & 1.02 & 1.11 \\ \hline
t=2 & 1.56 & 2.82 & 1.65 & 1.07 & 1.12 \\ \hline
\end{tabular}
\caption{The lower bounds for $\theta\geq1$ given by Lemma \ref{LemmaBrazilChinaUSPronetoSavings} for $\phi=0.2$.}
\label{table3}
\end{table}
In order to have common parameters for all countries, we adopt, henceforth, $\theta=\theta_{\tau}=2.82$. Before moving forward with the analysis, it is instructive to observe that the utility maximization problem of household $h\in G_{t}$, $t\geq0$, is normally written as
\begin{eqnarray*}
    \textrm{max}_{c\in\mathbb{R}^{3}_{+}}& u(c_{t},c_{t+1},c_{t+2}) \\
    \textrm{ s.t. }& (p_{t},p_{t+1},p_{t+2})\cdot (c-e^{h})=0.
\end{eqnarray*}
for $(p_{t},p_{t+1},p_{t+2})\in\mathbb{R}^{3}_{++}$, although it can also be written sequentially as
\begin{eqnarray*}
    \textrm{max}_{c\in\mathbb{R}^{3}_{+}}& u(c_{t},c_{t+1},c_{t+2}) \\
    \textrm{ s.t. }& c_{t}=e_{t}-s_{t}\\
    & c_{t+1}=e_{t}+\frac{p_{t}}{p_{t+1}}s_{t}-s_{t+1}\\
    & c_{t+2}=\phi e_{t}+\frac{p_{t+1}}{p_{t+2}}s_{t+1},
\end{eqnarray*}
so that the return rates $r_{t}=p_{t}/p_{t+1}$, $t\geq0$, bear a direct economic interpretation in terms of the returns on real savings. For instance, the graphic below illustrates the return rates implied by Lemma \ref{lemmaTailProneToSavings} for different values of $a_{3}\in(-1,1/\alpha_{\tau}\vert\lambda_{3}\vert)\approx(-1,2.43)$.

\begin{figure}[h]
\includegraphics[width=11cm]{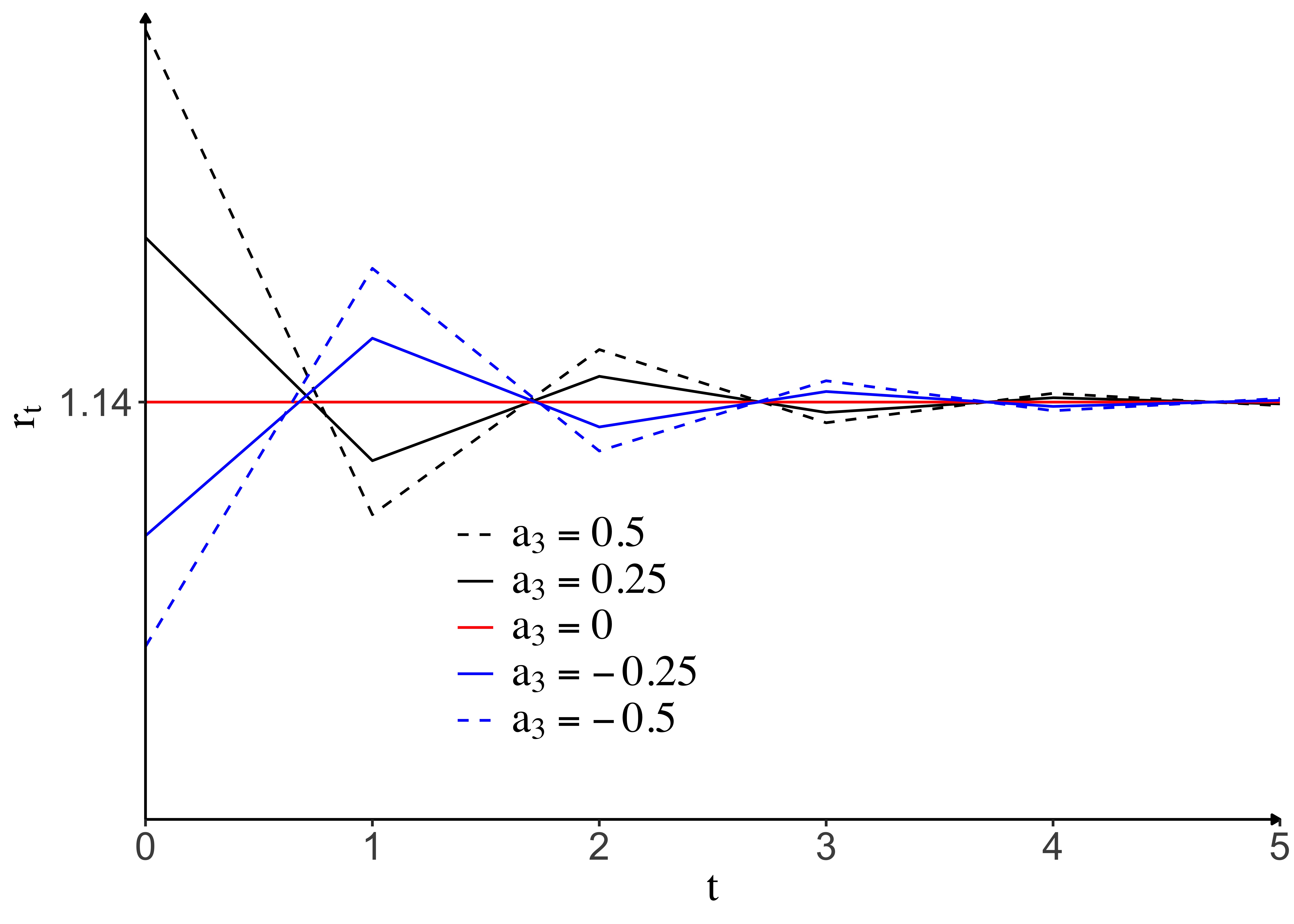}
\centering
\caption{Optimal return rates given by Corollary \ref{corHPO} in the tail economy of Brazil (i.e., $\theta_{\tau}=2.82$ and $\alpha_{\tau}=\gamma_{4}=1.14$) for different values of $a_{3}\in(-1,1/\alpha_{\tau}\vert\lambda_{3}\vert)=(-1,2.43)$.}
\label{fig2}
\end{figure}

Every value of $a_{3}\in(-1,1/\alpha_{\tau}\vert\lambda_{3}\vert)\approx(-1,2.43)$ furnishes a possible optimal return rate sequence in Figure \ref{fig2}. Our next step, therefore, is to depart from these optimal return rates (i.e., optimal prices) and calculate backward the remaining ones according to market clearing equations (\ref{eqSystemEqEquations}). The solutions are illustrated in Figure \ref{fig3}.

\begin{figure}[h]
\includegraphics[width=11cm]{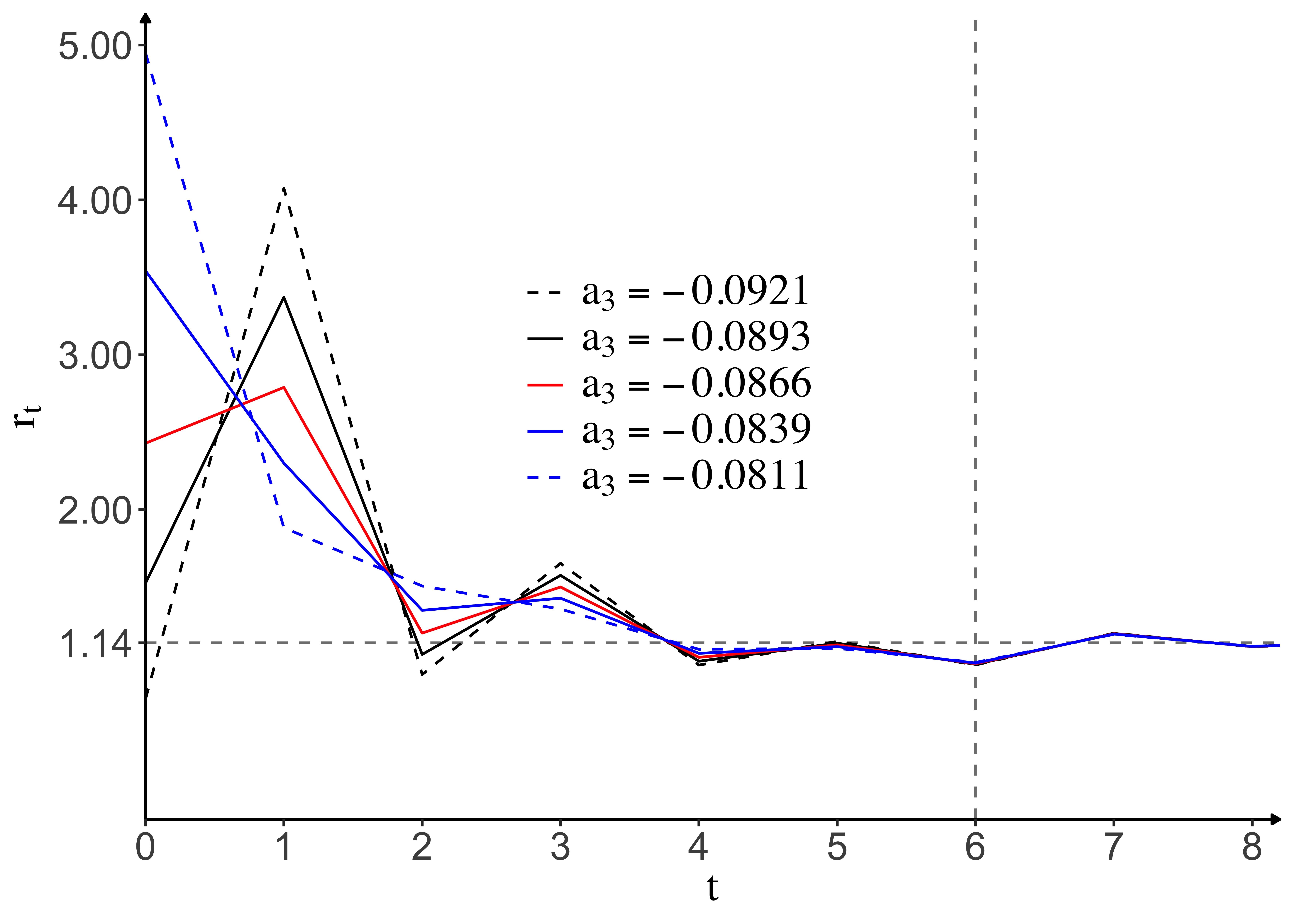}
\centering
\caption{Optimal return rates in the ``complete'' economy of Brazil for different values of $a_{3}\in(-1,1/\alpha_{\tau}\vert\lambda_{3}\vert)=(-1,2.43)$. The vertical dashed line marks the beginning of the tail economy.}
\label{fig3}
\end{figure}

First, notice that although Figure \ref{fig2} depicts solutions for $a_{3}\in(-1,1/\alpha_{\tau}\vert\lambda_{3}\vert)$, the backward calculation algorithm does not reach its end (i.e., calculates $p_{0}$) for all these values. This is because Theorem \ref{theoAlgorithmFinal} implies that the algorithm will reach its end for at least one Pareto optimal equilibrium of the tail economy, not necessarily for all of them. Figure \ref{fig3} reveals that for $a_{3}\in[-0.0963,-0.0827]$, the algorithm does calculate all prices and, therefore, all return rates. 

In order to select one of these return rate sequences to design our social security system, we can simply pick the one that minimizes the standard deviation of $\{r_{t}\}_{0\leq t\leq 6}$ (these are all the return rates faced by generations $H_{t}$, $0\leq t\leq 5$, i.e., generations that are not in the tail of the economy). Clearly, this is not the only choice criterion, but it seems a reasonable one for our theoretical purposes. This point will be further discussed later on. For Brazil, this minimum-variance sequence is given by $a_{3}=-0.0895$ (depicted by the solid red line in Figure \ref{fig3}).

Figure \ref{fig4} uses the same methodology to calculate backward these ``minimum variance'' return rates for all countries.

\begin{figure}[h]
\includegraphics[width=11cm]{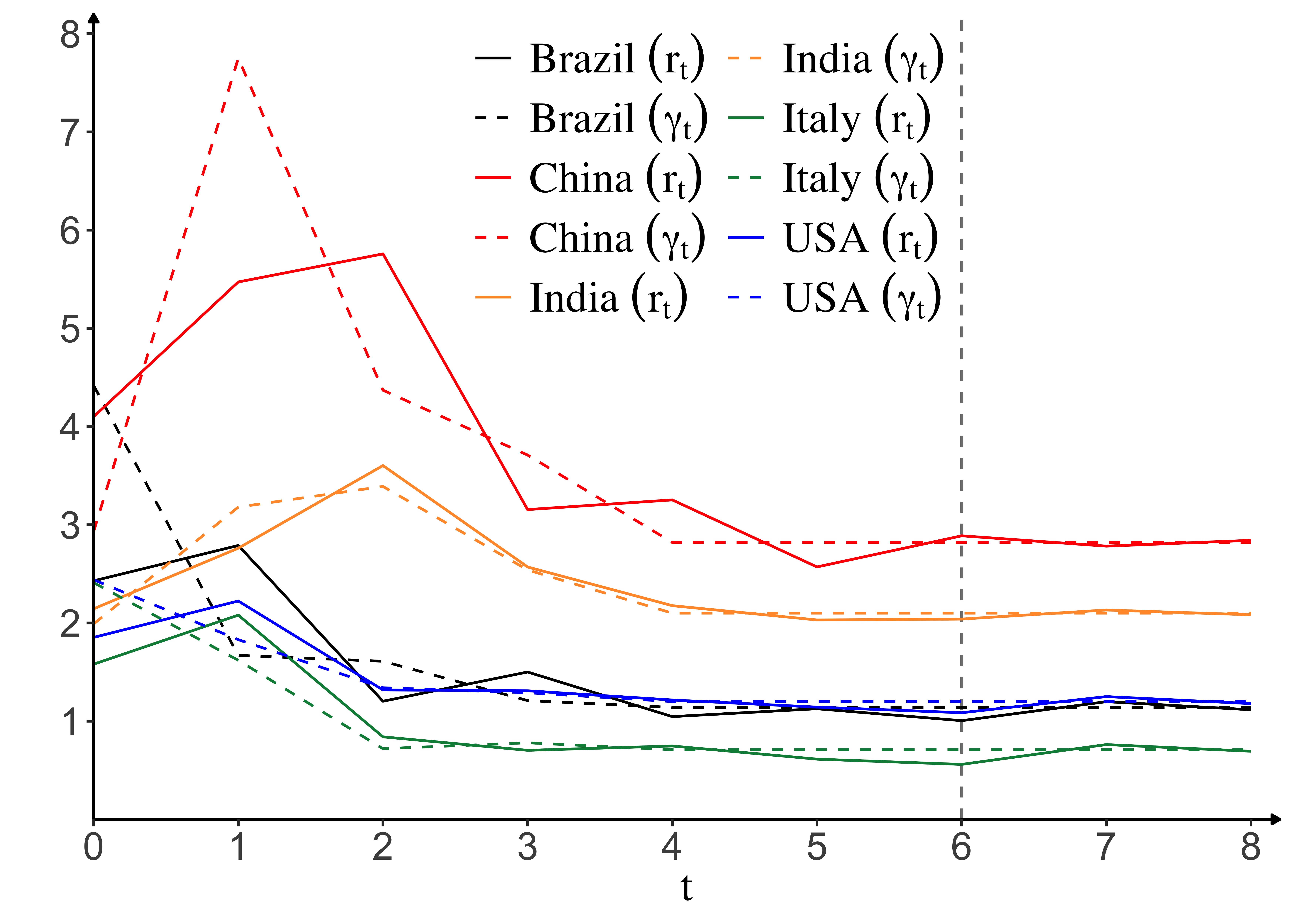}
\centering
\caption{Optimal return rates with minimum variance in the ``complete'' economy of Brazil, China, India, Italy and the United States, along with the values of $\{\gamma_{t}\}_{0\leq t\leq 8}$. The vertical dashed line marks the beginning of the tail economy.}
\label{fig4}
\end{figure}

Figure \ref{fig4} reveals that, as a rule-of-thumb, one can expect the return rates to roughly follow the growth rate of the aggregate endowment of the economy. When productivity is constant and, therefore, the aggregate endowment follows the demography, this is what \textcite[pp. 471-474]{Samuelson_1958} called a ``biological interest rate'' . In general, this can be seen as a pattern linking GDP growth rates and the optimal return rates of pay-as-you-go social security systems. 

The optimal rates, however, deviate from this pattern once a country faces a demographic transition (i.e., significant changes in demographic growth rates). This can be more clearly seen in the cases of Brazil and China when considering the minimum variance sequences of Figure \ref{fig4}.    

Finally, the return rates allow us to find the optimal contribution and replacement rates for each generation, since the design of the social security system is given by
\begin{eqnarray*}
    s^{t}&=&x^{t}(p_{t},p_{t+1},p_{t+2})-(e_{t},e_{t},\phi e_{t})\\
    &=&\frac{e_{t}}{1+\theta+\theta^{2}}\biggr(\frac{r_{t+1}+\phi}{r_{t}r_{t+1}}-\theta(1+\theta),\frac{r_{t}r_{t+1}+\phi}{r_{t+1}}-\theta(1+\theta),\theta^{2}r_{t+1}(1+r_{t})-\phi(1+\theta)\biggr),
\end{eqnarray*}
for $0\leq t\leq 7$. Therefore, the contribution rates in the first period of life of generation $G_{t}$, $0\leq t\leq 7$, are
\begin{eqnarray*}
    \sigma^{t}_{1}=\frac{1}{1+\theta+\theta^{2}}\biggr(\frac{r_{t+1}+\phi}{r_{t}r_{t+1}}-\theta(1+\theta)\biggr).
\end{eqnarray*}

Clearly, the design described above is only illustrative and uses a highly simplified setup. Nonetheless, to create a full-scale model, we would still need to follow this section step-by-step, albeit with some adaptations. 

For instance, each period would usually be a single calendar-year, the population would be heterogeneous (both in preferences and in endowments), the set of Pareto optimal equilibria of the tail economy would be given by a more general version of Corollary \ref{corHPO} and long-term forecasts would be used. Notice that in this section, the tail economy starts in 2070 (which would be fairly good if we were in 1950, but is far from ideal since we are already in 2026), and we simply let $\gamma_{t}=\gamma_{4}$, $t\geq4$, to ease the calculations.

There are still four remarks that must be highlighted. First, this design method relies on households' preferences, and this is a fundamental property of it. Traditionally, Governments are mostly interested in the fiscal balance of pay-as-you-go systems; therefore, when reforming pay-as-you-go systems, considerations regarding insured households' preferences are (if they exist) secondary.

The second remark is the following. When choosing among the multiple optimal equilibria given by the backward calculation algorithm, we adopted a minimum-variance criterion. There are, however, other possible criteria. One could, for instance, choose the return rates sequence that minimizes the distance to the ongoing contribution and replacement rates of the pay-as-you-go system, or select the one that minimizes changes in the replacement rates of the current older-generations.

The third remark is that one may use the backward calculation algorithm departing from the set of Pareto optimal equilibria from a tail economy even if one does not have a sufficient condition to guaranty that the ``complete'' economy is prone-to-savings. For instance, we could use a value of $\theta\geq1$ that does not satisfy Lemma \ref{LemmaBrazilChinaUSPronetoSavings}. However, in this case, there may be no value in the tail of the economy that can lead to a valid return rate sequence (i.e., there may be no solution for (\ref{eqSystemEqEquations}) for any value of prices in the tail economy).

The last remark is that since the endowments are exogenous, the retirement age also is. In Section \ref{sec2}, the social security system was completely defined by the contribution and the replacement rate of each household because, with a two-period lifespan, one can only retire in the last one. With $S\geq3$, however, the retirement age becomes a central parameter of the system. Therefore, as far as this paper goes, the backward calculation algorithm can only provide optimal contribution and replacement rates of the social security system for a \textit{given retirement age value}, which is implicit in the exogenous endowments.

\section{Concluding remarks}\label{sec5}

The main contribution of this paper is to reveal how general equilibrium theory and, particularly, overlapping generations models can be used to design optimally balanced pay-as-you-go social security systems in a general demographic setup. 

Since this design is based on the backward calculation algorithm in \textcite{Dognini_2025}, one fundamental aspect of it is the definition of well-behaved tail-economies and the full characterization of their sets of Pareto optimal equilibria. Corollary \ref{corHPO} goes in this direction, but in order to deploy this design in a real-world setup, we still need a broader result. Another milestone that still must be achieved is a more general sufficient condition for an economy to be prone-to-savings (see Lemma \ref{LemmaBrazilChinaUSPronetoSavings}).

Nonetheless, this paper sets forth a theoretical basis and a ``topological justification'' for the design of pay-as-you-go social security system through overlapping generations economies with well-behaved tails.

\appendix
\section*{Appendix}\label{appx}

All proofs are stated in this \hyperref[appx]{Appendix}.

\begin{proof}[Proof of Proposition~{\upshape\ref{propOptReturnRates1Dim}}]
Let $I=\{r_{1}\in\mathbb{R}_{++}\mid \exists p\in\mathbb{R}^{\infty}_{++}\textrm{ monetary equilibrium s.t. } r_{1}=p_{1}/p_{2}$\}. A standard analysis of the overlapping generations economy with fiat money $\mathcal{E}$ allows us to state that $I\neq\emptyset$ (e.g., \textcite{BalaskoCassShell_1980,OkunoZilcha_1981}). Notice also that: (i) if $r_{1}\in I$ and $0<s_{1}<r_{1}$, then $s_{1}\in I$; (ii) $\sup I\in I$; and (iii) if $r^{\alpha}_{1},r^{\beta}_{1}\in I$, $r^{\alpha}_{1}>r^{\beta}_{1}$, then $r^{\alpha}_{t}>r^{\beta}_{t}$, $t\geq1$. In particular, (i) and (ii) imply $I=(0,\tilde{r}_{1}]$ and (iii) implies that the allocation defined by $\tilde{r}_{1}>0$ Pareto dominates all other monetary equilibria, since households utility is strictly increasing according to the respective equilibrium price rate. I proceed with the following claim.
\begin{claim}\label{claimCassCriterion}
Let $p\in\mathbb{R}^{\infty}_{++}$ be a monetary equilibrium. Then, $\sum_{t\geq0} 1/H_{t}p_{t}<+\infty$ if, and only if, $r_{n}<\alpha_{\min}$, for some $n\geq1$.
\end{claim}
Since $p\in\mathbb{R}^{\infty}_{++}$ is a monetary equilibrium, the market clearing equations imply 
\begin{eqnarray*}
\phi(r_{t})=\frac{1}{p_{1}\alpha_{0}}\prod^{t-1}_{i=1}\frac{r_{i}}{\alpha_{i}}=\frac{H_{0}}{H_{t}p_{t}},
\end{eqnarray*}
for $t\geq1$. If $\sum_{t\geq0} 1/H_{t}p_{t}<+\infty$, then $\lim_{t\rightarrow\infty}1/H_{t}p_{t}=0$. Therefore, $\lim_{t\rightarrow\infty}\phi(r_{t})=\lim_{t\rightarrow\infty}H_{0}/H_{t}p_{t}=0$, so that $\lim_{t\rightarrow\infty}r_{t}=\phi^{-1}(\lim_{t\rightarrow\infty}\phi(r_{t}))=0$. We conclude that $r_{n}<\alpha_{\min}$, for some $n\geq1$.

For the converse, let $n\geq1$ be such that $r_{n}<\alpha_{\min}/(1+\varepsilon)$, for some $\varepsilon>0$. Since $\phi(\cdot)$ is strictly increasing, we have
\begin{eqnarray*}
    r_{n}<\frac{\alpha_{\min}}{1+\varepsilon} \implies \phi(r_{n+1})=\frac{r_{n}\phi(r_{n})}{\alpha_{n}} \leq \frac{r_{n}\phi(r_{n})}{\alpha_{\min}}< \frac{\phi(r_{n})}{1+\varepsilon}.
\end{eqnarray*}
Notice also this last inequality implies $\phi (r_{n+1})<\phi(r_{n})< \phi(\alpha_{\min}/(1+\varepsilon))$, so that $r_{n+1}<\alpha_{\min}/(1+\varepsilon)$. By induction, we have $\phi(r_{n+i})< \phi(r_{n})/(1+\varepsilon)^{i}$, $i\geq1$. Therefore,
\begin{eqnarray*}
    \sum_{t\geq0} \frac{1}{H_{t}p_{t}}=\sum^{n-1}_{t\geq0} \frac{1}{H_{t}p_{t}} + \frac{1}{H_{0}}\sum_{t\geq n} \phi(r_{t})< \sum^{n-1}_{t\geq0} \frac{1}{H_{t}p_{t}} + \frac{\phi(r_{n})}{H_{0}}\sum_{i\geq 0}\frac{1}{(1+\varepsilon)^{i}}<\infty. 
\end{eqnarray*}
This discussion also allows us to state the following.
\begin{claim}\label{claimLim0}
    Let $p\in\mathbb{R}^{\infty}_{++}$ be a monetary equilibrium. Then $\lim_{n\rightarrow\infty}r_{t}=0$ if, and only if, $r_{n}<\alpha_{\min}$, for some $n\geq1$.
\end{claim}
I proceed with the claim below.
\begin{claim}\label{claimLMinusLPlus}
There are $0<L_{-}<L_{+}$ such that if $L\geq L_{+}$, then $f_{t}(L)\leq L$, $t\geq1$; and if $L\leq L_{-}$, then $f_{t}(L)\geq L$, $t\geq1$.
\end{claim}
Notice that we can take $L_{+}=\psi(\alpha_{\max})$, since $L\geq L_{+}$ implies
\begin{eqnarray*}
    f_{t}(L)=\psi(\alpha_{t}\phi(L))\leq \psi(\alpha_{\max})=L_{+}\leq L,
\end{eqnarray*}
for $t\geq1$. Also, let $L_{-}\leq 1$, be any value such that
\begin{eqnarray*}
    \frac{\alpha_{\min}}{2}\biggr(\frac{1+\sqrt{1+8/(\alpha_{\min}L_{-})}}{2}\biggr)\geq1.
\end{eqnarray*}
Then, $L\leq L_{-}$ implies
\begin{eqnarray*}
    f_{t}(L)=\psi\biggr(\frac{\alpha_{t}L}{1+L}\biggr)\geq \psi\biggr(\frac{\alpha_{\min}L}{2}\biggr)=\frac{L\alpha_{\min}}{2}\biggr(\frac{1+\sqrt{1+8/(\alpha_{\min}L)}}{2}\biggr)\geq L,
\end{eqnarray*}
for $t\geq1$.
Next, for $L>0$, $1\leq t\leq n$, let $r^{n}_{t}(L)=f_{t}\circ\ldots\circ f_{n}(L)$. Notice that
\begin{eqnarray*}
    r^{n}_{t}=f_{t}(r^{n}_{t+1})=\psi(\alpha_{t}\phi(r^{n}_{t+1}))\implies r^{n}_{t}\phi(r^{n}_{t})=\alpha_{t}\phi(r^{n}_{t+1}),
\end{eqnarray*}
for $1\leq t<n$. Therefore, if $\lim_{n\rightarrow\infty}r^{n}_{t}=\tilde{r}_{t}$, $t\geq1$, then
\begin{eqnarray*}
\tilde{r}_{t}\phi(\tilde{r}_{t})=\alpha_{t}\phi(\tilde{r}_{t+1}),
\end{eqnarray*}
for $t\geq1$.

If $L\leq L_{-}$, then Claim \ref{claimLMinusLPlus} implies
\begin{eqnarray*}
    r^{n+1}_{t}(L)=f_{t}\circ\ldots\circ f_{n}\circ f_{n+1}(L)\geq f_{t}\circ\ldots\circ f_{n}(L)=r^{n}_{t}(L)
\end{eqnarray*}
for $1\leq t\leq n$, and $\{r^{n}_{t}(L)\}_{n\geq t}$, $t\geq1$, is a non-decreasing sequence. Since $f_{t}(\cdot)$, $t\geq1$, is bounded, there exists $\lim_{n\rightarrow\infty}r^{n}_{t}(L)=\tilde{r}_{t}(L)$, with
\begin{eqnarray}\label{eqLowerBound1}
    \tilde{r}_{t}(L)=\lim_{n\rightarrow\infty}r^{n}_{t}(L)\geq r^{t}_{t}(L)=f_{t}(L)\geq \psi(\alpha_{\min}\phi(L))>0,
\end{eqnarray}
for $t\geq1$. Since the lower bound in (\ref{eqLowerBound1}) is valid for all $t\geq1$, Claim \ref{claimLim0} implies that 
\begin{eqnarray}\label{eqLowerBound2}
    \tilde{r}_{t}(L)\geq\alpha_{\min},
\end{eqnarray}
for $t\geq1$. Let $\tilde{p}(L)\in\mathbb{R}^{\infty}_{++}$ be the monetary equilibrium defined by $\{\tilde{r}_{t}(L)\}_{t\geq1}$. Claim \ref{claimCassCriterion}, (\ref{eqLowerBound2}) and the \textcite{Cass_1972} criterion applied to the allocation defined by $\tilde{p}(L)\in\mathbb{R}^{\infty}_{++}$ \parencite[Theorems 3A and 3B, pp. 803-804]{OkunoZilcha_1980} imply that it is Pareto optimal. Since the allocation defined by $\tilde{r}_{1}$ Pareto dominates all others defined by monetary equilibria, we conclude that $\tilde{r}_{1}(L)=\tilde{r}_{1}$, for $0<L\leq L_{-}$.

If $L\geq L_{+}$, then Claim \ref{claimLMinusLPlus} implies
\begin{eqnarray*}
    r^{n+1}_{t}(L)=f_{t}\circ\ldots\circ f_{n}\circ f_{n+1}(L)\leq f_{t}\circ\ldots\circ f_{n}(L)=r^{n}_{t}(L)
\end{eqnarray*}
for $1\leq t\leq n$, and $\{r^{n}_{t}(L)\}_{n\geq t}$, $t\geq1$, is a non-increasing sequence. Therefore, there exists $\lim_{n\rightarrow\infty}r^{n}_{t}(L)=\tilde{r}_{t}(L)$, with
\begin{eqnarray}\label{eqLowerBound3}
    \tilde{r}_{t}(L)=\lim_{n\rightarrow\infty}r^{n}_{t}(L)\geq \lim_{n\rightarrow\infty}r^{n}_{t}(L_{-})\geq \alpha_{\min}>0,
\end{eqnarray}
for $t\geq1$, where the previous to last inequality is due to (\ref{eqLowerBound2}). Let $\tilde{p}(L)\in\mathbb{R}^{\infty}_{++}$ be the monetary equilibrium defined by $\{\tilde{r}_{t}(L)\}_{t\geq1}$. Claim \ref{claimCassCriterion}, (\ref{eqLowerBound3}) and the \textcite{Cass_1972} criterion applied to the allocation defined by $\tilde{p}(L)\in\mathbb{R}^{\infty}_{++}$ \parencite[Theorems 3A and 3B, pp. 803-804]{OkunoZilcha_1980} imply that it is Pareto optimal. Since the allocation defined by $\tilde{r}_{1}$ Pareto dominates all others defined by monetary equilibria, we conclude that $\tilde{r}_{1}(L)=\tilde{r}_{1}$, for $L\geq L_{+}$.

Finally, let $L_{-}<L<L_{+}$ and notice that
\begin{eqnarray*}
    \tilde{r}_{1}=\lim_{n\rightarrow\infty}r^{n}_{1}(L_{-})\leq \lim_{n\rightarrow\infty}r^{n}_{1}(L)\leq \lim_{n\rightarrow\infty}r^{n}_{1}(L_{+})=\tilde{r}_{1} \implies \lim_{n\rightarrow\infty}r^{n}_{1}(L)=\tilde{r}_{1}.
\end{eqnarray*}
\end{proof}

\begin{proof}[Proof of Corollary~{\upshape\ref{corMarginalDemChange}}]
Notice that Proposition \ref{propOptReturnRates1Dim} implies
\begin{eqnarray*}
    \frac{\partial \tilde{r}_{1}}{\partial \alpha_{t}}&=&\frac{\partial}{\partial \alpha_{t}}\biggr(\lim_{n\rightarrow\infty} f_{1}\circ \ldots \circ f_{n}(\alpha)\biggr)\\
    &=&\frac{\partial}{\partial \alpha_{t}}\biggr(f_{1}\circ \ldots \circ f_{t}\circ \biggr(\lim_{n\rightarrow\infty} f_{t+1}\circ\ldots\circ f_{n}(\alpha)\biggr)\biggr)\\
    &=&\frac{\partial}{\partial \alpha_{t}}\biggr(f_{1}\circ \ldots \circ \psi(\alpha_{t}\phi(\alpha))\biggr)\\
    &=&(f^{\prime}(\alpha))^{t-1}\psi^{\prime}(\alpha\phi(\alpha))\phi(\alpha),
\end{eqnarray*}
for $t\geq1$. Next, we have
\begin{eqnarray*}
    \psi^{\prime}(r)=\frac{1}{2}\biggr(1+\frac{r+2}{\sqrt{r^{2}+4r}}\biggr)
\end{eqnarray*}
for $r>0$, so that
\begin{eqnarray*}
    \psi^{\prime}(\alpha\phi(\alpha))=\frac{1}{2}\biggr(1+\frac{\alpha^{2}/(1+\alpha)+2}{\sqrt{\alpha^{4}/(1+\alpha)^{2}+4\alpha^{2}/(1+\alpha)}}\biggr)
    =\frac{(\alpha+1)^{2}}{\alpha(\alpha+2)}.
\end{eqnarray*}
Also, $f^{\prime}(r)=\psi^{\prime}(\alpha\phi(r))\alpha\phi^{\prime}(r)$, $r>0$, so that $f^{\prime}(\alpha)=1/(\alpha+2)$.We conclude that
\begin{eqnarray*}
    \frac{\partial \tilde{r}_{1}}{\partial \alpha_{t}}=\frac{1+\alpha}{(2+\alpha)^{t}},
\end{eqnarray*}
for $t\geq1$.
\end{proof}

\begin{proof}[Proof of Lemma~{\upshape\ref{lemmaTailProneToSavings}}]
Notice that (\ref{eqRealSavingsTail}) implies that
\begin{eqnarray*}
    \sum_{h\in\tilde{G}_{t}}\frac{\tilde{s}^{h}(\tilde{p}_{t},\tilde{p}_{t+1})}{\tilde{H}^{\tau k}_{t}}       &\geq&\frac{(1+\theta+\alpha)e}{(1+\theta+\theta^{2})(1+\alpha)}\biggr(\frac{\theta^{2}}{1+\theta+\alpha}-\frac{\Vert \tilde{p}_{t+1}\Vert}{\Vert \tilde{p}_{t}\Vert}\biggr)\\
    &=&f_{2}(\alpha,\theta,e)\biggr(f_{1}(\alpha,\theta)-\frac{\Vert \tilde{p}_{t+1}\Vert}{\Vert \tilde{p}_{t}\Vert}\biggr),
\end{eqnarray*}
for $(\tilde{p}_{t},\tilde{p}_{t+1})\in\mathbb{R}^{4}_{++}$, $t\geq0$. By assumption, $f_{1}(\alpha,\theta)>\alpha^{-2}$, so that there are $\varepsilon,\delta>0$ such that
\begin{eqnarray*}
    \frac{f_{2}(\alpha\theta,e)}{f_{1}(\alpha,\theta)f_{2}(\alpha,\theta,e)-\delta} \leq \frac{\alpha^{2}}{1+\varepsilon}.
\end{eqnarray*}
Then, we have
\begin{eqnarray*}
    \sum_{h\in\tilde{G}_{t}}\frac{\tilde{s}^{h}(\tilde{p}_{t},\tilde{p}_{(t+1)})}{\tilde{H}_{t}}\leq\delta&\implies&f_{2}(\alpha,\theta,e)\biggr(f_{1}(\alpha,\theta)-\frac{\Vert \tilde{p}_{t+1}\Vert}{\Vert \tilde{p}_{t}\Vert}\biggr)\leq \delta\\
    &\implies& \frac{\Vert \tilde{p}_{t}\Vert}{\Vert \tilde{p}_{t+1}\Vert}\leq\frac{f_{2}(\alpha,\theta,e)}{f_{1}(\alpha,\theta)f_{2}(\alpha,\theta,e)-\delta} \leq \frac{\alpha^{2}}{1+\varepsilon}=\frac{\tilde{\alpha}^{\tau k}_{t}}{1+\varepsilon},
\end{eqnarray*}
for $t\geq0$, and the economy $\tilde{\mathcal{T}}_{k}$ is prone to savings. 
\end{proof}

\begin{proof}[Proof of Proposition~{\upshape\ref{propTailSetofEquilibria}}]
Suppose $\tilde{p}\in\tilde{\mathcal{H}}_{\tau k}$, so that $(\tilde{p}_{0},\tilde{p}_{1})\in\mathcal{B}_{0}(\sigma_{0})$. The market clearing equations for $\tilde{\mathcal{T}}_{k}$, $k\geq0$, can be written as
\begin{eqnarray*}
p_{t+4}=\frac{-\theta^{2}p_{t}-\theta(\theta+\alpha)p_{t+1}+\Delta(\alpha,\theta)p_{t+2}-\alpha(\theta+\alpha)p_{t+3}}{\alpha^{2}}
\end{eqnarray*}
for $t\geq0$, with $\Delta(\alpha,\theta)=(1+\alpha+\alpha^{2})(1+\theta+\theta^{2})-\alpha^{2}-\alpha\theta-\theta^{2}$.  Let $\Omega(\alpha,\theta)\in\mathbb{R}^{4\times 4}$ be the matrix
\begin{eqnarray*}
    \Omega(\alpha,\theta)=\begin{bmatrix}
        0 & 1 & 0 & 0\\
        0 & 0 & 1 & 0\\
        0 & 0 & 0 & 1\\
        \frac{-\theta^{2}}{\alpha^{2}} & \frac{-\theta(\theta+\alpha)}{\alpha^{2}}
        & \frac{\Delta(\alpha,\theta)}{\alpha^{2}} & \frac{-\alpha(\theta+\alpha)}{\alpha^{2}}
    \end{bmatrix},
\end{eqnarray*}
so that
\begin{eqnarray}\label{eqRecurrence1}
    \begin{bmatrix}
        p_{t}\\
        p_{t+1}\\
        p_{t+2}\\
        p_{t+3}
    \end{bmatrix}=\Omega(\alpha,\theta)^{t}
    \begin{bmatrix}
        p_{0}\\
        p_{1}\\
        p_{2}\\
        p_{3}
    \end{bmatrix},
\end{eqnarray}
for $t\geq0$. The characteristic polynomial of $\Omega(\alpha,\theta)$ is given by
\begin{eqnarray*}
    p(\lambda)=\frac{1}{\alpha}(\lambda-\theta)\biggr(\lambda-\frac{1}{\alpha}\biggr)\biggr(\alpha\lambda^{2}+(1+\alpha)(1+\theta)\lambda+\theta\biggr),
\end{eqnarray*}
Therefore, the eigenvalues of $\Omega(\alpha,\theta)$ are $\lambda_{1}=\theta$, $\lambda_{2}=1/\alpha$,
\begin{eqnarray*}
    \lambda_{3}&=&\frac{-(1+\alpha)(1+\theta)+\sqrt{(1+\alpha)^{2}(1+\theta)^{2}-4\theta\alpha}}{2\alpha}\\
    \lambda_{4}&=&\frac{-(1+\alpha)(1+\theta)-\sqrt{(1+\alpha)^{2}(1+\theta)^{2}-4\theta\alpha}}{2\alpha},
\end{eqnarray*}
and the corresponding eigenvectors $\omega_{i}$, $1\leq i\leq 4$, can be written as
\begin{eqnarray*}
    \omega_{i}=\begin{bmatrix}
        1\\
        \lambda_{i}\\
        \lambda_{i}^{2}\\
        \lambda_{i}^{3}
    \end{bmatrix},
\end{eqnarray*}
for $1\leq i\leq 4$. Notice first that $\lambda_{4}<\lambda_{3}<0<1/\alpha<\theta$, and so $\{\omega_{i}\}_{1\leq i\leq 4}$ is a basis of $\mathbb{R}^{4}$. Then, we can always write $(p_{0},p_{1},p_{2},p_{3})=\sum^{4}_{i=1}a_{i}\omega_{i}$, so that
(\ref{eqRecurrence1}) can be simplified to
\begin{eqnarray}\label{eqRecurrence2}
        \begin{bmatrix}
            p_{t} \\ p_{t+1} \\ p_{t+2} \\ p_{t+3}
        \end{bmatrix} = \sum^{4}_{i=1} a_{i}\lambda_{i}^{t} \begin{bmatrix}
            1 \\ \lambda_{i} \\ \lambda^{2}_{i} \\ \lambda^{3}_{i}
        \end{bmatrix}
\end{eqnarray}
for $t\geq0$. Next, notice that
\begin{eqnarray*}
    \vert\lambda_{3}\vert=-\lambda_{3}<\frac{1}{\alpha}&\iff& (1+\alpha)(1+\theta)-\sqrt{(1+\alpha)^{2}(1+\theta)^{2}-4\theta\alpha}<2\\
    &\iff& (1+\alpha)^{2}(1+\theta)^{2}-4(1+\alpha)(1+\theta)+4<(1+\alpha)^{2}(1+\theta)^{2}-4\theta\alpha\\
    &\iff&0<\alpha+\theta,
\end{eqnarray*}
and that
\begin{eqnarray*}
    \vert\lambda_{4}\vert=-\lambda_{4}>\theta&\iff&(1+\alpha)(1+\theta)+\sqrt{(1+\alpha)^{2}(1+\theta)^{2}-4\theta\alpha}>2\alpha\theta\\
    &\iff&(1+\alpha)^{2}(1+\theta)^{2}-4\theta\alpha>4\alpha^{2}\theta^{2}-4\alpha\theta(1+\alpha)(1+\theta)+(1+\alpha)^{2}(1+\theta)^{2}\\
    &\iff&\alpha+\theta>0.
\end{eqnarray*}
Since $p_{t}>0$, $t\geq0$, we have that $\vert \lambda_{4}\vert>\theta>1/\alpha>\vert \lambda_{3}\vert$ and $\lambda_{4}<0$ imply that $a_{4}=0$. Otherwise, a negative price would eventually arise in the sequence. Also, $\theta>1/\alpha>\vert \lambda_{3}\vert$ implies, by the same reasoning, $a_{1}\geq0$. Then, we can write $(p_{0},p_{1},p_{2},p_{3})=\sum^{3}_{i=1}a_{i}\omega_{i}$ and looking at the first coordinate yields $a_{1}+a_{2}+a_{3}=1$. We conclude from (\ref{eqRecurrence2}) that
    \begin{eqnarray}\label{eqRecurrence3}
        \begin{bmatrix}
            p_{t} \\ p_{t+1} \\ p_{t+2} \\ p_{t+3}
        \end{bmatrix} = 
        a_{1} \theta^{t} \begin{bmatrix}
            1 \\ \theta \\ \theta^{2} \\ \theta^{3}
        \end{bmatrix}+
        a_{2} \alpha^{-t} \begin{bmatrix}
            1 \\ 1/\alpha \\ 1/\alpha^{2} \\ 1/\alpha^{3}
        \end{bmatrix}+
        a_{3} \lambda^{t}_{3} \begin{bmatrix}
            1 \\ \lambda_{3} \\ \lambda_{3}^{2} \\ \lambda^{3}_{3}
        \end{bmatrix},
    \end{eqnarray}
for $t\geq0$. The converse is immediate.
\end{proof}
\begin{proof}[Proof of Lemma~{\upshape\ref{LemmaBrazilChinaUSPronetoSavings}}]
Let $\tilde{G}_{t}=G_{2t}\cup G_{2t+1}$, $0\leq t\leq 2$. Then, real savings per capita of generation $\tilde{G}_{t}$, $0\leq t\leq 2$, in this two-period model is given by
   \begin{eqnarray*}
        \sum_{h\in\tilde{G}_{t}}\frac{\tilde{s}^{h}(\tilde{p}_{t},\tilde{p}_{t+1})}{\tilde{H}_{t}}&=&\sum_{h\in G_{2t}}\frac{\tilde{p}_{t}\cdot(\tilde{e}^{h}_{t}-\tilde{x}^{h}_{t}(\tilde{p}_{t},\tilde{p}_{(t+1)}))}{\Vert \tilde{p}_{t}\Vert (H_{2t}+H_{2t+1})}+\sum_{h\in G_{2t+1}}\frac{\tilde{p}_{t}\cdot(\tilde{e}^{h}_{t}-\tilde{x}^{h}_{t}(\tilde{p}_{t},\tilde{p}_{(t+1)}))}{\Vert \tilde{p}_{t}\Vert(H_{2t}+H_{2t+1})}\\
        &=&\frac{e_{2t}}{1+\alpha_{2t}}-\frac{(1+\theta)(\tilde{p}_{t},\tilde{p}_{t+1})\cdot (e_{2t},e_{2t},\phi e_{2t},0)}{(1+\theta+\theta^{2})\Vert \tilde{p}_{t}\Vert(1+\alpha_{2t})}+\ldots\\
        &&\ldots \frac{\alpha_{2t} \tilde{p}_{t2}e_{2t+1}}{\Vert \tilde{p}_{t}\Vert(1+\alpha_{2t})}-\frac{\alpha_{2t} (\tilde{p}_{t},\tilde{p}_{t+1})\cdot (0,e_{2t+1},e_{2t+1},\phi e_{2t+1})}{(1+\theta+\theta^{2})\Vert \tilde{p}_{t}\Vert(1+\alpha_{2t})}\\
        &\geq&\frac{e_{2t}}{(1+\theta+\theta^{2})(1+\alpha_{2t})}\biggr(\theta^{2}-((1+\theta)\phi+\gamma_{2t})\frac{\Vert \tilde{p}_{t+1}\Vert}{\Vert\tilde{p}_{t}\Vert}\biggr).
    \end{eqnarray*}
Since $\gamma_{2t}<\theta^{2}\tilde{\alpha}_{t}-(1+\theta)\phi$, $0\leq t\leq 2$, there are $\varepsilon,\delta>0$ such that
\begin{eqnarray*}
    0<\frac{e_{2t}((1+\theta)\phi+\gamma_{2t})}{\theta^{2}e_{2t}-(1+\theta+\theta^{2})(1+\alpha_{2t})\delta} \leq \frac{\tilde{\alpha}_{t}}{1+\varepsilon}
\end{eqnarray*}
for $0\leq t\leq 3$. Then, we have
\begin{eqnarray*}
    \sum_{h\in\tilde{G}_{t}}\frac{\tilde{s}^{h}(\tilde{p}_{t},\tilde{p}_{t+1})}{\tilde{H}_{t}}\leq\delta&\implies&\frac{e_{2t}}{(1+\theta+\theta^{2})(1+\alpha_{2t})}\biggr(\theta^{2}-((1+\theta)\phi+\gamma_{2t})\frac{\Vert \tilde{p}_{t+1}\Vert}{\Vert \tilde{p}_{t}\Vert}\biggr)\leq \delta\\
    &\implies& \frac{\Vert \tilde{p}_{t}\Vert}{\Vert \tilde{p}_{t+1}\Vert}\leq\frac{e_{2t}((1+\theta)\phi+\gamma_{2t})}{\theta^{2}e_{2t}-(1+\theta+\theta^{2})(1+\alpha_{2t})\delta} \leq \frac{\tilde{\alpha}_{t}}{1+\varepsilon},
\end{eqnarray*}
for $0\leq t\leq 2$. 
\end{proof}

\printbibliography
\end{document}